\newif\ifUSENIX
\USENIXtrue

\newif\ifANON
\ANONtrue

\newif\ifextra
\extrafalse %

\ifUSENIX
  \documentclass[letterpaper,twocolumn,10pt]{article}

\usepackage{hyperref}
\hypersetup{
     colorlinks   = false,
     citecolor    = blue,
     linkcolor    = blue,
     urlcolor     = blue,
     hidelinks    = true,
}
\usepackage{amsthm}
\usepackage{amsmath,amsfonts,amssymb}
\usepackage{color}

\usepackage{qtree}
\usepackage{enumitem}

\setlist{noitemsep}
\usepackage{scalerel}

\usepackage{color,colortbl}
\definecolor{Gray}{gray}{0.9}

\usepackage[boxruled,titlenotnumbered,noend]{algorithm2e}

\usepackage{booktabs}

\newcommand{\zo}{\{0,1\}}

\renewcommand{\vec}[1]{\mathbf{#1}}
\newcommand{\mat}[1]{\vec{#1}}
\newcommand{\Y}{\mat{Y}}
\newcommand{\y}{\vec{y}}
\newcommand{\X}{\mat{X}}
\newcommand{\x}{\vec{x}}
\newcommand{\z}{\vec{z}}
\newcommand{\Z}{\mat{Z}}
\renewcommand{\c}{\vec{c}}
\renewcommand{\subset}{\subseteq}

\newcommand{\bg}{{\mathsf{big}}}
\newcommand{\sz}{{\mathsf{size}}}
\newcommand{\sm}{{\mathsf{sml}}}
\newcommand{\Bern}{\mathsf{Bern}}
\newcommand{\Norm}{\mathsf{N}}
\newcommand{\Binom}{\mathsf{Bin}}
\newcommand{\Pois}{\mathsf{Pois}}

\newcommand{\e}{\vec{e}}

\newcommand{\Att}{\mathcal{U}}
\renewcommand{\hat}[1]{\widehat{#1}}

  \newcommand{\bigmid}{\;\big\vert\;}
  \newcommand{\Bigmid}{\;\Big\vert\;}

  \newcommand{\indic}[1]{\mathbb{I}(#1)}
  
  \newcommand{\adv}{\mathsf{A}}
  \newcommand{\A}{\adv}
  
  \newcommand{\bbR}{\mathbb{R}}
  \newcommand{\bbN}{\mathbb{N}}

  \newcommand{\eps}{\varepsilon}
  \newcommand{\negl}{\mathrm{negl}}

\newtheorem{theorem}{Theorem}[section]
\newtheorem{claim}[theorem]{Claim}

\newtheorem{observation}[theorem]{Observation}
\newtheorem{example}[theorem]{Example}

\theoremstyle{definition} %
\newtheorem{definition}{Definition}[section]

  \newcommand{\cnt}{\EA}
  \newcommand{\EA}{\mathsf{EA}}
  \newcommand{\ambig}{\mathsf{amb}}
  \newcommand{\UEA}{\mathsf{EA_\ambig}}

  \def\E{\operatorname*{\mathbb{E}}}
  \def\poly{\mathop{\rm{poly}}\nolimits}

  \newcounter{notecounter}

  \usepackage{usenix-2020-09}
\else
  \documentclass[11pt]{article}
  \usepackage{fullpage}
  
\fi
\renewcommand{\c}{\vec{c}}

\begin{document}
\title{Attacks on Deidentification's Defenses}

\author{{\rm Aloni Cohen}\thanks{
We thank Kobbi Nissim for many helpful discussions; Ryan Sullivan for early discussions on EdX; and Gabe Kaptchuk, anonymous reviewers, and especially Mayank Varia for generous feedback.
This work was primarily done at Boston University's Hariri Institute of Computing and School of Law.
This work was supported by the DARPA SIEVE program under Agreement No. HR00112020021 and the National Science Foundation under Grant Nos. CNS-1915763 and SaTC-1414119. Any opinions, findings, and conclusions or recommendations expressed in this
material are those of the author and do not reflect the views of our funders.}  \\ University of Chicago }

\maketitle

\begin{abstract}
  Quasi-identifier-based deidentification techniques (QI-deidentification) are widely used in practice, including $k$-anonymity, $\ell$-diversity, and $t$-closeness.
  We present three new attacks on QI-deidentification: two theoretical attacks and one practical attack on a real dataset. In contrast to  prior work, our theoretical attacks work even if every attribute is a quasi-identifier. Hence, they apply to $k$-anonymity, $\ell$-diversity, $t$-closeness, and most other QI-deidentification techniques.

  First, we introduce a new class of privacy attacks called \emph{downcoding attacks}, and prove that every QI-deidentification scheme is vulnerable to downcoding attacks if it is minimal and hierarchical.
  Second, we convert the downcoding attacks into powerful \emph{predicate singling-out (PSO)} attacks, which were recently proposed as a way to demonstrate that a privacy mechanism fails to legally anonymize under Europe's General Data Protection Regulation.
  Third, we use LinkedIn.com to reidentify 3 students in a $k$-anonymized dataset published by EdX (and show  thousands are potentially vulnerable), undermining EdX's claimed compliance with the Family Educational Rights and Privacy Act.

  The significance of this work is both scientific and political.
  Our theoretical attacks demonstrate that QI-deidentification may offer no protection even if every attribute is treated as a quasi-identifier.
  Our practical attack demonstrates that even deidentification experts acting in accordance with strict privacy regulations fail to prevent real-world reidentification.
  Together, they rebut a foundational tenet of QI-deidentification and challenge the actual arguments made to justify the continued use of $k$-anonymity and other QI-deidentification techniques.
\end{abstract}

\ifUSENIX

\fi

\section{Introduction}
\label{sec:intro}

Quasi-identifier-based deidentification (QI-deidentification) is widely used in practice.
The most well known QI-deidentification techniques are is $k$-anonymity \cite{SamaratiS98}.
Throughout this work we usually speak about $k$-anonymity specifically, but everything applies without modification to $\ell$-diversity \cite{MKGV07}, $t$-closeness \cite{t-closeness}, and many other QI-deidentification refinements.

A relatively small number of data points suffice to distinguish individuals from the general population. For example, in the 2010 census 44\% of the population was unique based only on census block, age, and sex \cite{AbowdDeclaration}.
Turning this insight into a privacy notion, $k$-anonymity aims to capture a sort of anonymity of a crowd.

A data release is $k$-anonymous if any individual row in the release cannot be distinguished from $k-1$ other individuals in the release using certain attributes called \emph{quasi-identifiers}.
Quasi-identifiers are sets of attributes that are potentially available to an attacker from other sources, combinations of which may uniquely distinguish an individual within the dataset.
$k$-anonymity requires that the {equivalence class} of every record---the set of records with identical quasi-identifiers---is of size at least $k\ge 2$. A common choice for $k$ is $5$.\footnote{For example, U.S.\ Department of Education's FAQ on disclosure avoidance states that ``statisticians consider a cell size of 3 to be the absolute minimum although larger minimums (e.g., 5 or 10) may be used to further mitigate disclosure risk'' (\url{https://studentprivacy.ed.gov/resources/frequently-asked-questions-disclosure-avoidance}).  Based on this language, EdX chose $k=5$.}
$\ell$-diversity, $t$-closeness, and many QI-deidentification techniques refine $k$-anonymity in the sense that they collapse to $k$-anonymity when every attribute is treated as a quasi-identifier (Sec.~\ref{sec:QI-survey}).

Real world reidentification attacks, including on the Netflix and AOL datasets \cite{narayanan2008robust, barbaro_zeller_2006}, led to a policy debate about the QI-deidentification. Critics argued that the distinction between quasi-identifying attributes and other attributes---foundational to the whole approach---was untenable \cite{BrokenPromises, myths-and-fallacies}. Defenders argued that deidentification experts are good at determining what information is externally available \cite{cavoukian2014identification, ontario2014big}.
The debate left unspoken and unexamined the core tenet of QI-deidentification: that if \emph{every} attribute is treated as a quasi-identifier, then $k$-anonymity provides meaningful protection. Our work is the first to directly challenge that tenet.

\paragraph{Motivation}
Why bother attacking QI-deidentification?
After all, the security and privacy research communities don't put much stock in these techniques. For example, it is well known that contrived mechanisms can formally satisfy $k$-anonymity but provide no protection.
Even so, many policymakers and practitioners are convinced that QI-deidentification is effective in the real world.

Our goal in this paper is to rebut the actual arguments that QI-deidentification practioners use to justify its continued use.
We rebut three arguments that---until this work---have gone unchallenged. First, that no attacks have been shown against datasets deidentified by experts and in accordance with strict privacy regulations, let alone simple attacks. Second, that $k$-anonymity provides meaningful protection when every attribute is a quasi-identifier.
Third, that although QI-deidentification doesn't meet cryptographic standards of security, it suffices to meet the obligations in data protection regulation.
We briefly elaborate these three arguments next.

Rhetorically, trust in QI-deidentification hinges on the wholesale dismissal of existing attacks as unconvincing.
Practitioners dismiss many attacked datasets as ``improperly de-identified'' \cite{cavoukian2014identification}.
``Proper de-identification'' must be done by a ``statistical expert'' and in accordance with procedures outlined in regulation \cite{ElEmamSystematic},
the increasing availability of QI-deidentification software notwithstanding.
This argument has proven very effective in policy spheres.
Moreover, practitioners dismiss attacks carried out by privacy researchers  \emph{because} they are privacy researchers. That these attacks are published in ``research based articles within the highly specialized field of computer science'' is used to argue that re-identification {requires} a ``highly skilled `expert'\,'' and therefore is of little concern \cite{ontario2014big}.

Technically, trust in QI-deidentification hinges on an unspoken, unexamined tenet:
\begin{quote}
	\textit{QI-deidentification's tenet: If \emph{every} attribute is treated as quasi-identifying, then $k$-anonymity provides meaningful protection.}
\end{quote}
Treating every attribute as quasi-identifying defines away one major critique of QI-deidentification---namely, that the ex ante categorization of attributes as quasi-identifying or not is untenable and reckless. Moreover, when all attributes are quasi-identifying, the distinctions among $k$-anonymity, $\ell$-diversity, and $t$-closeness collapse (Section~\ref{sec:QI-survey}). Prior attacks against $k$-anonymity fail in this setting.

Legally, the use of QI-deidentification hinges on the gap between the protection required by regulation and the protection desired by the academic research community.
Practitioners claim only that QI-deidentification meets regulatory standards, not security researchers' stringent standards.
For example, cryptographic security definitions typically make no assumptions about the techniques or auxiliary knowledge available to an adversary. However, the European Union's General Data Protection Regulation (GDPR) restricts the adversary's techniques by protecting only against ``means reasonably likely to be used'' by an attacker.\footnote{GDPR, Article 4} Likewise, the United State's Family Educational Rights and Privacy Act (FERPA) restricts the adversary's knowledge by protecting only against an attacker lacking ``personal knowledge of the relevant circumstances.''\footnote{34 CFR \S99.3}

\paragraph{Contributions}
We present three attacks on QI-deidentification schemes: two theoretical attacks and one real world reidentification attack. Together, these attacks undermine the above justifications for the continued use of QI-deidentification.

First, we introduce a new class of privacy attack called \emph{downcoding}, which recovers large fractions of the data hidden by QI-deidentification without any auxiliary knowledge.
In short, downcoding undoes hierarchical generalization.
A downcoding attack takes as input a dataset generalized and recovers some fraction of the generalized data.
We call this downcoding as it corresponds to recoding records down a generalization hierarchy.

We prove that every QI-deidentification scheme is vulnerable to downcoding attacks if it is \emph{minimal} and \emph{hierarchical}.
QI-deidentification is hierarchical if it works by generalizing attributes according to a fixed hierarchy (e.g., city$\to$country$\to$continent). QI-deidentification is minimal if no record is generalized more than necessary to achieve the privacy requirement, in a weak, local sense.
Our downcoding attacks are powered by a simple observation: \emph{minimality leaks information}.
Figures~\ref{fig:downcoding-example} and~\ref{fig:minimality} give simple examples of downcoding and of leakage from minimality, respectively.

Second, we convert our downcoding attacks into powerful \emph{predicate singling-out (PSO)} attacks.
PSO attacks were recently proposed as a way to demonstrate that a privacy mechanism fails to legally anonymize under the GDPR \cite{altman2020hybrid, CohenNissimSinglingOut}.
We introduce a stronger type of PSO attack called \emph{compound} PSO attacks and prove that  minimal hierarchical QI-deidentification enables compound PSO attacks, greatly improving over the prior work.

Our downcoding and PSO attacks are the first attacks on QI-deidentification that work even when every attribute is a quasi-identifier. As such, they apply to QI-deidentification beyond $k$-anonymity, and refute the foundational tenet of QI-deidentification.

Third, we used LinkedIn.com to reidentify 3 students in a $k$-anonymized dataset published by Harvard and MIT from their online learning platform EdX. Despite being ``properly'' $k$-anonymized by ``statistical experts'' in accordance with FERPA, we show that thousands more  students are potentially vulnerable to reidentification and disclosure.

Not only do these attacks rebut the arguments described above, they also show that QI-deidentification fails to satisfy three properties of a worthwhile measure of privacy of a computation, even without resorting to contrived mechanisms. Namely, we show that QI-deidentification mechanisms used in practice aren't robust to post-processing, do not compose, and rely on distributional assumptions on the data for their security.

\paragraph{Organization}
Section~\ref{sec:intro:related} discusses related work.
Section~\ref{sec:prelims} introduces notation and defines $k$-anonymity, along with hierarchical and minimal $k$-anonymity. Section~\ref{sec:downcoding} defines downcoding attacks and proves that minimal hierarchical $k$-anonymous mechanisms enable them.
Section~\ref{sec:pso} defines compound predicate singling-out attacks and proves that minimal hierarchical $k$-anonymous mechanisms enable them.
Section~\ref{sec:edx} describes the EdX dataset and shows that it is vulnerable to reidentification.
Section~\ref{sec:conclusion} concludes that our attacks rebut the three core arguments that support the continued use of QI-deidentification in practice.
The appendix includes additional details and proofs.

\section{Related Work}
\label{sec:intro:related}

Samarati and Sweeney proposed $k$-anonymity for statistical disclosure limitation in 1998 \cite{SamaratiS98, samarati2001protecting, sweeney2002k}.
As new attacks were discovered, $k$-anonymity gave rise to more refined QI-deidentification techniques including $\ell$-diversity, $t$-closeness, and many others (below).

Samarati was the first to study minimality for $k$-anonymity \cite{samarati2001protecting}. Our downcoding attacks build on prior work on {minimality attacks}~\cite{wong2009anonymization, cormode2010minimizing}. These works demonstrate that minimality can be used to infer sensitive attributes and violate $\ell$-diversity, but not $k$-anonymity. They introduce two defenses against their attacks.
  One is yet another refinement of $k$-anonymity called $m$-confidentiality \cite{wong2009anonymization}.
	The second claims that certain anonymization algorithms offer protection for free (i.e., ``methods which only inspect the QI attributes to determine the [equivalence classes]'') \cite{cormode2010minimizing}.
	In contrast, we use minimality to downcode, a new attack that violates $k$-anonymity itself and that defeats both defenses from  prior work.

	Predicate singling-out (PSO) attacks were recently introduced  in the context of data anonymization under Europe's General Data Protection Regulation (GDPR) \cite{CohenNissimSinglingOut}.
	They were proposed as a mathematical test to show that a privacy mechanism fails to legally anonymize data under Europe's General Data Protection Regulation (GDPR) \cite{altman2020hybrid, CohenNissimSinglingOut}. The prior work gives a simple but weak PSO attack against a large class of $k$-anonymous mechanisms. We give much stronger PSO attacks against a restricted class of $k$-anonymous mechanisms.

	Prior work shows that $k$-anonymity does not compose: multiple $k$-anonymous datasets can completely violate privacy when combined \cite{ganta2008composition}. We show for the first time that composition failures can occur in real world uses of $k$-anonymity.

 	Differential privacy (DP) \cite{DMNS06} presents one alternative to QI-deidentification, especially DP synthetic data \cite{pate-gan} or local DP \cite{local-DP}. Switching to DP requires accepting that the resulting data will not provide the one-to-one correspondence with underlying records that makes QI-deidentification so attractive to users and laypeople.

			\subsection{Syntactic de-identification beyond $k$-anonymity}
			\label{sec:QI-survey}

			We reviewed the deidentification definitions included in the most comprehensive survey we could find \cite{k_anon_refinements_survey}.
			Our downcoding attacks apply to any \emph{refinement} of $k$-anonymity: namely, any definition that collapses to $k$-anonymity when every attribute is quasi-identifying. These include:
			\begin{itemize}
				\item $k$-anonymity and variants: $k^m$-, $(\alpha,k)$-, $p$-sensitive-, $(k,p,q,r)$-, and $(\epsilon,m)$-anonymity
				\item $\ell$-diversity and variants: entropy-, recursive-, disclosure-recursive, multi-attribute-, $\ell^+$-, and $(c,\ell)$-diversity
				\item $t$-closeness and variant $(n,t)$-closeness
				\item $m$-invariance, $m$-confidentiality
			\end{itemize}
			Our downcoding attacks don't apply to Anatomy (which doesn't generalize quasi-identifiers at all) or differential privacy (which eschews the quasi-identifier framework all together).
			We have not determined whether the following definitions -- which bound some posterior probability given the deidentified dataset -- refine $k$-anonymity in the relevant sense: $\delta$-presence, $\eps$-privacy, skyline privacy, $(\rho_1,\rho_2)$-privacy, $(c,k)$-safety, and $\rho$-uncertainty.
			
			We leave testing our downcoding attacks on actual deidentification software packages for future work.
		Free to use software packages include
		ARX Anonymization, 		$\mu$-Argus,		sdcMicro,		University of Texas Toolkit, 		Amnesia,		Anonimatron,		Python Mondrian.		All but Python Mondrian implement hierarchical algorithms. ARX Anonymization, sdcMicro, and Amenesia offer some version of local recoding (footnote~\ref{note:local-recoding}). To the best of our knowledge, none guarantee  minimality.

		\begin{figure*}[ht]
		\begin{center}
			\begin{tabular}{ccc}

		$\X = $ \begin{tabular}{crc}
			\toprule
			ZIP & Income & COVID\\
			\midrule
			91010 & \$125k & Yes \\

			91011 & \$105k & No \\

			91012 & \$80k & No \\

			20037 & \$50k & No \\

			20037 & \$20k & No \\

			20037 & \$25k & Yes \\
			\bottomrule
			\end{tabular}
		&
		$\Y = $ \begin{tabular}{ccc}
			\toprule
			ZIP & Income & COVID\\
			\midrule
			9101$\star$ & \$75--150k & $\star$ \\

			9101$\star$ & \$75--150k & $\star$\\

			9101$\star$ & \$75--150k & $\star$ \\

			20037 & \$0--75k & $\star$ \\

			20037 & \$0--75k & $\star$ \\

			20037 & \$0--75k & $\star$\\
			\bottomrule
		\end{tabular}
		&
		$\Z = $ \begin{tabular}{ccc}
			\toprule
			ZIP & Income & COVID\\
			\midrule
			91010 & \$125--150k & $\star$ \\

			9101$\star$ & \$100--125k & $\star$ \\

			9101$\star$ & \$75--150k & $\star$ \\

			20037 & \$0--75k & No \\

			20037 & \$0--75k & $\star$ \\

			20037 & \$25k & Yes \\
			\bottomrule
		\end{tabular}
		\end{tabular}
		\end{center}
		\caption{\label{fig:downcoding-example}An example of downcoding. $\Y$ is a minimal hierarchical $3$-anonymized version of $\X$ (treating every attribute as part of the quasi-identifier and leaving the generalization hierarchy implicit). $\Z$ is a downcoding of $\Y$: it generalizes $\X$ and strictly refines $\Y$.}
		\end{figure*}

			\begin{figure}[ht]
			\begin{center}
				\begin{tabular}{ccc}
			\begin{tabular}{cc}
				\toprule
				Old & Rich \\
				\midrule
				1 & 1 \\

				0 & 0 \\

				1 & 0 \\

				0 & 0 \\
				\bottomrule
				\end{tabular}
			&
			\begin{tabular}{cc}
				\toprule
				 Old & Rich \\
				\midrule
				\rowcolor{Gray}
				$\bigstar_{\scaleto{1}{3.5pt}}$ & $\bigstar_{\scaleto{5}{3.5pt}}$ \\

				$\bigstar_{\scaleto{2}{3.5pt}}$ & $\bigstar_{\scaleto{6}{3.5pt}}$ \\

				\rowcolor{white}
				$\bigstar_{\scaleto{3}{3.5pt}}$ & 0 \\

				$\bigstar_{\scaleto{4}{3.5pt}}$ & 0 \\
				\bottomrule
				\end{tabular}
			&
			\begin{tabular}{cc}
				\toprule
				Old & Rich \\
				\midrule
				\rowcolor{Gray}
				1 & $\bigstar_{\scaleto{7}{3.5pt}}$ \\

				\rowcolor{white}
				0 & 0 \\

				\rowcolor{Gray}
				1 & $\bigstar_{\scaleto{8}{3.5pt}}$ \\

				\rowcolor{white}
				0 & 0 \\
				\bottomrule
				\end{tabular}
			\end{tabular}
			\end{center}
			\caption{\label{fig:minimality} An example of minimality and inferences from minimality. Attributes are binary and $\bigstar = \{0,1\}$. The middle and right datasets are both minimal hierarchical $2$-anonymous versions of the left dataset with respect to $Q =$ \{Old, Rich\}. The right dataset is also globally optimal: it generalizes as few attributes as possible.
			Minimality implies that every pair of redacted entries in the same column in matching rows must contain both a 0 and 1. Hence,
				$\{\bigstar_{\scaleto{1}{3.5pt}},\bigstar_{\scaleto{2}{3.5pt}}\}=\{\bigstar_{\scaleto{3}{3.5pt}},\bigstar_{\scaleto{4}{3.5pt}}\}=\{\bigstar_{\scaleto{5}{3.5pt}},\bigstar_{\scaleto{6}{3.5pt}}\}=\{\bigstar_{\scaleto{7}{3.5pt}},\bigstar_{\scaleto{8}{3.5pt}}\}=\{0,1\}$,
			allowing  downcoding. Only one bit of information does not follow directly from minimality of the middle table: whether or not $\bigstar_1 = \bigstar_5$.}
			\end{figure}

\section{Preliminaries}
\label{sec:prelims}

\subsection{Notation}
\label{sec:prelims:notation}

Generally, fixed parameters are denoted by capital letters (e.g., number of dimensions $D$) and indices use the corresponding lowercase letter (e.g., $d = 1, \ldots, D$).
For $a,b\in \bbN$, let $[a,b] = \{a,a+1,\ldots, b\}$ and $[b] = [1,b]$.

$\Att^D$ is a $D$-dimensional data universe, where $\Att$ is the attribute domain. For simplicity we take all attribute domains to be identical, though in reality they are usually distinct (e.g., the EdX dataset).

A \emph{record} $\x = (x_1,\dots,x_D)$ is an element of the data universe.
A \emph{generalized record} $\y$, denoted $(y_1,\dots,y_D)$, is a subset of the data universe specified by the Cartesian product $y_1\times \dots \times y_D$, where $y_d\subset \Att$ for every $d \in [D]$.
Note that a record $\x$ naturally corresponds to the generalized record $(\{x_1\}, \dots, \{x_D\})$, a singleton.
We say $\y$ \emph{generalizes} $\x$ if $\x \in \y$ (i.e., $\forall d,\ x_d \in y_d$).
For example, $\y = (Female, 1970\text{--}1975)$ generalizes $\x=(Female, 1972)$.
For generalized records $\z\subset \y$, we say that $\y$ \emph{generalizes} $\z$ and $\z$ \emph{refines} $\y$.
If $\z\subsetneq \y$, the generalization/refinement is \emph{strict}.

A \emph{dataset} $\X$ is an $N$-tuple of records $(\x_1,\dots,\x_N)$.
$\X$ can be viewed as a matrix with $\X_{n,d}$ the $d$th coordinate of $\x_n$.
A \emph{generalized dataset} $\Y$ is an $N$-tuple of generalized records $(\y_1,\dots,\y_N)$.
For (generalized) datasets $\Y, \Z$, we write $\Z\preceq \Y$ if $\z_{n} \subset \y_{n}$ for all $n$. We extend the meaning of generalization and refinement accordingly. We write $\Z\prec\Y$ when at least one containment is strict.
We call $\y_{n}$ the record in $\Y$ \emph{corresponding to} $\z_n$, and vice-versa.
Note that $\preceq$ is a partial order on datasets of $N$ records from a given data universe.\footnote{%
More generally, we could consider datasets whose rows are permuted relative to one another. Define $\Z \preceq \Y$ if there exists a permutation $\pi:[N]\to[N]$ such that $\z_n \subset \y_{\pi(n)}$ for all $n$, choosing some canonical $\pi$ arbitrarily if more than one exists. Then $\preceq$ is a partial order over equivalence classes of datasets induced by $\Y \sim \Y' \iff \exists \pi\ \forall n\ \y_n = \y'_{\pi(n)}$. We omit this additional complexity for clarity. We believe all our results would hold, mutatis mutandis.}

\subsection{$k$-anonymity}
\label{sec:k-anon}

Formally, $\Y$ is $k$-anonymous if any individual row in the release cannot be distinguished from $k-1$ other individuals \cite{SamaratiS98}.
This requirement is typically parameterized by a subset $Q$ of the attribute domains $Q\subset \{\Att_d\}_{d\in [D]}$ called a \emph{quasi-identifier}.
We denote by $\y(Q)$ the restriction of $\y$ to $Q$. For $\Y = (\y_1, \dots, \y_N)$, we denote by $I(\Y,\y, Q) \triangleq \{n: \y_n(Q) = \y(Q)\}$ the indices of records in $\Y$ that match $\y$ on $Q$ (including $\y$ itself).
Let $\cnt(\Y,\y,Q) = |I(\Y,\y,Q)|$. This is called the \emph{effective anonymity} of $\y$ in $\Y$ with respect to $Q$.

\begin{definition}[$k$-anonymity]
	\label{def:k-anon}
	For $k\ge 2$, $\Y$ is \emph{$k$-anonymous} with respect to $Q$ if for all $\y \in \Y$, $\cnt(\Y,\y,Q) \ge k$.
	An algorithm $M:\X \mapsto \Y$ is \emph{$k$-anonymizer} if for every $\X$, $\Y \gets M(\X)$ is $k$-anonymous \emph{(anonymity)} and generalizes $\X$ \emph{(correctness)}.
	We omit $Q$ when $Q = \Att^D$ is the whole data universe.
\end{definition}

A few remarks are in order.
First, beyond correctness and anonymity, $k$-anonymity places no restriction on the output $\Y$.
Second, the term quasi-identifier is inconsistently defined in the literature.
Our definition of a quasi-identifier as the collection of multiple attributes is from Sweeney \cite{sweeney2002k}.
Quasi-identifier is commonly used to refer to one of the constituent attributes---including by the authors of the EdX dataset \cite{edX-dataset}.
So each \cite{sweeney2002k}-quasi-identifier consists of multiple \cite{edX-dataset}-quasi-identifers.
We adopt the quasi-identifier-as-a-set definition because it simplifies the discussion of the EdX dataset in Section~\ref{sec:edx}.
The distinction disappears in Sections~\ref{sec:downcoding} and~\ref{sec:pso}: our downcoding and PSO attacks work even when every attribute is part of the quasi-identifier (i.e., $Q = \Att^D$).

\subsubsection{Hierarchical $k$-anonymity}

It is easy to contrive $k$-anonymizers that reveal $\X$ completely.
Directing our attention to more natural and widespread mechanisms, we focus on \emph{hierarchical} $k$-anonymizers.

A common way of $k$-anonymizing data is to {generalize} an attribute domain $\Att$ according to a data-independent \emph{generalization hierarchy} $H$ which specifies how a given attribute may be \emph{recoded}.\footnotemark~
Many natural ways of generalizing data fits this mold: using nesting geographies (e.g., city$\to$state$\to$country); dropping digits of postal codes (e.g., $91011\to9101\star \to 910\star\star$); grouping ages into ranges of 5, 10, 25, or 50 years; suppressing attributes or whole records altogether; and the techniques used to create the EdX dataset.

	\footnotetext{
	\label{note:local-recoding}
	Hierarchical algorithms differ on whether they use \emph{local recoding} or \emph{global recoding}.
	Using local recoding, attributes in different records can be generalized to different levels of the hierarchy. Using global recoding, all records must use the same level in the hierarchy for any given attribute.
	We consider local recoding which produces higher quality datasets in general.}

Formally, a generalization hierarchy $H$ defines a structured collection of permissible subsets $\y$ of an attribute domain $\Att$ (Figure~\ref{fig:main-thm-hierarchy}).
$H$ is a rooted tree labelled by subsets of $\Att$, where the subsets on any level of $H$ form a partition of $\Att$ and the partition on every level is a strict refinement of the partition above. The label of the root is $\Att$ itself, and the leaves are all labelled with singletons $\{x\}$.
Identifying $H$ with the set of all its labels, we write $y \in H$ if there is some node in $H$ labelled by $y$. We extend the hierarchy $H$ to the data universe $\Att^D$ coordinate-wise, writing $\y \in H^D$ if $y_d \in H$ for all $d\in[D]$.

\begin{definition}[Hierarchical $k$-anonymity]
$\Y$ \emph{respects} $H$ if $\y\in H^D$ for all $\y \in \Y$.
An algorithm $M:(\X,H)\mapsto \Y$ is a \emph{hierarchical} $k$-anonymizer if for all $\X$ and all hierarchies $H$, $M_H:\X \mapsto M(\X,H)$ is a $k$-anonymizer and its output $\Y = M(\X,H)$ respects $H$.
\end{definition}

Observe that one can always implement hierarchical $k$-anonymity by simply outputting $N$ copies of $\Att^D$.
But a privacy technique that completely destroys the data is not useful, which leads us to consider data quality.

We consider \emph{minimal} mechanisms~\cite{samarati2001protecting}. A mechanism is minimal if no record is generalized more than necessary to achieve the privacy requirement (in a local way).
For example, suppose a $k$-anonymous $\Y$ contains a location attribute. If there is a subset of records whose location ``USA'' can be changed to ``California'' without violating $k$-anonymity, then the mechanism that produced $\Y$ would not be minimal.
 Anoter example is given in Figure~\ref{fig:minimality}.
We call this property \emph{minimality} because it is equivalent to requiring minimality with respect to the partial ordering $\preceq$.
Unlike global optimality, minimality is computationally tractable.

\begin{definition}[Hierarchical minimality]
$M:(\X,H)\mapsto\Y$ is \emph{minimal} if $\Y$ is always minimal in the set of all $H$-respecting, $k$-anonymous $\Y$ that generalize $\X$, partially ordered by $\preceq$. That is, for all strict refinements $\Z \prec \Y$, either: (a) $\Z$ is not $k$-anonymous, (b) $\Z$ does not respect $H$, or (c) $\Z$ does not generalize $\X$.
\end{definition}

\section{Downcoding attacks on syntactic privacy techniques}
\label{sec:downcoding}

We study a new class of attacks on hierarchical $k$-anonymity called \emph{downcoding attacks} and prove that all minimal hierarchical $k$-anonymizers are vulnerable to downcoding attacks.
Our downcoding attacks are powerful yet computationally straightforward.
The attacks apply as is to $\ell$-diversity, $t$-closeness, and the many QI-deidentification techniques in Section~\ref{sec:QI-survey}.
They demonstrate that even when every attribute is treated as a quasi-identifier, any privacy offered by QI-deidentification depends on unstated distributional assumptions about the dataset.

\subsection{Overview}

In short, downcoding undoes hierarchical generalization.
A downcoding attack takes as input a dataset generalized and recovers some fraction of the generalized data.
We call this downcoding as it corresponds to recoding records down a generalization hierarchy.
Our downcoding attacks are powered by a simple observation: \emph{minimality leaks information}.
Figures~\ref{fig:downcoding-example} and~\ref{fig:minimality} give simple examples of downcoding and of leakage from minimality, respectively.

	We prove that there exist data distributions and hierarchies such that every minimal hierarchical $k$-anonymizer is vulnerable to downcoding attacks.
		Hence any privacy provided by QI-deidentification is subject to distributional assumptions.

		The downcoding attack adversary $\A$ gets as input a QI-deidentified dataset $\Y$ which is the output of an unknown mechanism $M$ on an unknown dataset $\X$. $\A$ also knows anything published with $\Y$, namely $N$, $k$, and the hierarchy $H$. (Without $H$ data users would be unable to interpret $\Y$.) Finally we also allow the adversary to depend on the data distribution $U$. One interpretation is that the security that a mechanism affords against downcoding attacks depends on limiting the attacker's knowledge, which is not good security practice. Moreover, in many settings $U$ can be efficiently learned from an independent sample $\X'$.

	  Formally, we construct a distribution $U$ over $\omega(\log n)$ attributes and a generalization hierarchy $H$ such that every minimal hierarchical algorithm enables downcoding attacks on datasets drawn i.i.d.\ from $U$. Our first attack uses a natural data distribution (i.e., clustered heteroskedastic data in Section~\ref{sec:downcoding:example}) and a tree-based hierarchy, and allows an attacker to completely recover a constant fraction of the deidentified records with high probability. Our second attack uses a less natural data distribution and hierarchy, and allows an attacker to recover $3/8$ths of every record with 99\% probability.

		Even with the assumptions on $M$ and the knowledge of $\A$, our attacks are far more general that typical attacks against QI-deidentification.
		For example, the attacks that motivated $t$-closeness as a refinement of $\ell$-diversity don't even apply to a single well-defined mechanism~\cite{t-closeness}.
		They show only that it is possible for a mechanism to produce $\ell$-diverse outputs that are vulnerable. In contrast, we show attacks on a large and well-defined class of mechanisms. Moreover, our attacks work against all QI-deidentification definitions simultaneously, not any one alone.

\subsection{Definition}
Let $\Y$ be a $k$-anonymous version of a dataset $\X$ with respect to generalization hierarchy $H$.
A downcoding attack takes $\Y$ as input and outputs a strict refinement $\Z$ of $\Y$ that simultaneously respects $H$ and generalizes $\X$.

\begin{definition}[Downcoding attack]
	Let $\Y$ be a hierarchical $k$-anonymous generalization of a (secret) dataset $\X$ with respect to some hierarchy $H$.
	$\Z$ is a \emph{downcoding} of $\Y$ if $\X\preceq \Z$, $\Z\prec \Y$, and $\Z\in H$.
\end{definition}

\begin{observation}\label{obs:downcoding-k-anon}
	If $\Y$ is minimal and $\Z$ is a downcoding of $\Y$, then $\Z$ violates $k$-anonymity.
\end{observation}

We consider three measures of an attack's strength: How \emph{many} records are refined? How \emph{much} are records refined? How \emph{often} records refined?
Recall that if $\Z \prec \Y$, then $\z_n \subseteq \y_{n}$ for all $n$ and $\z_n \subsetneq \y_{n}$ for at least one $n$.

\begin{itemize}[leftmargin=*,itemindent=-10pt,label={}]
\item $\Delta_N$: How \emph{many} records are refined? For $\Delta_N \in \bbN$, we write $\Z \prec_{\Delta_N} \Y$ if there exist at least $\Delta_N$ distinct $n$ for which $\z_n \subsetneq \y_{n}$.
That is, $\Z$ strictly refines at least $\Delta_N$ records in $\Y$. An attacker prefers larger $\Delta_N$.

\item $\Delta_D$: How \emph{much} are the records refined? For $\Delta_D \in \bbN$, we write $\z \subsetneq_{\Delta_D} \y$ if there exist at least $\Delta_D$ distinct $d$ for which $z_d \subsetneq y_d$.
We write $\Z \prec_{\Delta_D} \Y$ if $\z_n\subsetneq\y_n \implies \z_n\subsetneq_{\Delta_D} \y_n.$
That is, either $\z_n=\y_{n}$ or it $\z_n$ strictly refines $\y_{n}$ along at least $\Delta_D$ dimensions.
	An attacker prefers larger $\Delta_D$.

 \item $\daleth$:\footnote{$\daleth$ is pronounced ``dah-let'' and is the fourth letter of the Hebrew alphabet.}
 How \emph{often} are records refined? Consider the probability experiment $\X \sim U^N$, $\Y\gets M(\X,H)$, and $\Z \gets \A(\Y)$ where $U$ is a distribution over data records, $M$ is a $k$-anonymizer, and $\A$ is a downcoding adversary.
	$\daleth(\Delta_N,\Delta_D) \in [0,1]$ is the probability that $\Z$ downcodes with parameters at least $\Delta_N$ and $\Delta_D$.
	For any fixed $\Delta_N$ and $\Delta_D$, an attacker prefers larger $\daleth$.
\end{itemize}

\subsection{Minimal $k$-anonymizers enable downcoding attacks}

Downcoding may seem impossible: How can one strictly refine $\Y$ using only the information contained in $\Y$ itself? Our attacks leverage \emph{minimality}. The mere fact that $\Y$ is a minimal hierarchical generalization of $\X$ reveals more information about $\X$ that we use for strong downcoding attacks. See Figure~\ref{fig:minimality} for a simple example.

A general-purpose hierarchical $k$-anonymizer $M$ works for every generalization hierarchy $H$. Our theorems state that there exist data distributions $U$ and corresponding hierarchies $H$ such that {every} minimal hierarchical $k$-anonymizer $M$ is vulnerable to downcoding. By Observation~\ref{obs:downcoding-k-anon}, these attacks defeat the $k$-anonymity of $M$.

\begin{theorem}\label{thm:weaker-downcoding}
For all $k \ge 2$, $D = \omega(\log N)$, there exists a distribution $U$ over $\mathbb{R}^D$, and a generalization hierarchy $H$ such that all minimal hierarchical $k$-anonymizers $M$ enable downcoding attacks with $\Delta_N = \Omega(N)$, $\Delta_D = 3D/8$, and $\daleth(\Omega(N),3D/8) > 1-\negl(N)$.
\end{theorem}

\begin{theorem}\label{thm:full-downcoding}
	For all constants $k \ge 2$, $\alpha > 0$, $D = \omega(\log N)$, and $T = \lceil N^2/\alpha \rceil$, there exists a distribution $U$ over $\Att^D = [0,T]^D$, and a generalization hierarchy $H$
	such that all minimal hierarchical $k$-anonymizers $M$
	enable downcoding attacks with $\Delta_N = N$, $\Delta_D = D$, and $\daleth(N,D) > 1-\alpha$.
	The attack also works for $k=N$ and $D = \omega(N \log N)$.
\end{theorem}

Each of the theorems has some advantages over the other. The attacker in Theorem~\ref{thm:full-downcoding} manages to recover {every} attribute of {every} record $\x\in\X$ except with probability $\alpha$. However the parameters of the construction depend polynomially on $1/\alpha$.
Theorem~\ref{thm:weaker-downcoding} removes this dependency, at the expense of attacking only a constant fraction of records and attributes---still a serious failure of $k$-anonymity.
The more significant advantage of Theorem~\ref{thm:weaker-downcoding} is that the data distribution and generalization hierarchy are both very natural (Example~\ref{ex:clustered-gaussians}). In contrast, the distribution and hierarchy in the proof of Theorem~\ref{thm:full-downcoding} are more contrived.

Full proofs of both Theorems~\ref{thm:weaker-downcoding} and~\ref{thm:full-downcoding} are in Appendix~\ref{app:proofs}.
Both proofs follow the same structure at a very high level. We prove a structural result on minimal, hierarchical $k$-anonymous mechanisms for a specially constructed hierarchy $H$ (Claims~\ref{claim:minimality-clustered} and \ref{claim:minimal-M-structural}). This structural result states that if $\X$ satisfies certain conditions then $\Y$ must take a restricted form which allows the downcoding adversary to construct $\Z$. To prove the theorem, we construct a data distribution $U$ such that random $\X\sim U^N$ will satisfy the conditions of the structural result with probability close to 1.

\subsection{Example: Clustered Gaussians}
\label{sec:downcoding:example}
The proof of Theorem~\ref{thm:weaker-downcoding} shows that distributions satisfying certain properties are vulnerable to downcoding attacks. Example~\ref{ex:clustered-gaussians} describes a family of clustered Gaussian distributions that satisfy those properties. Here we give an instantiation of this family of distributions for $k=10$ and describe the corresponding hierarchy and downcoding adversary.

We sample $N=100$ records $\x$ i.i.d.\ as follows. Pick $\sz = \bg$ with probability $1/10$, and $\sz = \sm$ otherwise. Pick a cluster $t \in \{1,\dots,10\}$ uniformly at random. Sample each attribute of $\x$ i.i.d.\ from the cluster centered at $c_t = 130t$ depending on $\sz$:
 If $\sz = \sm$ sample from $\mathrm{N}(c_t, 1)$ distribution. If $\sz = \bg$ sample from the $\mathrm{N}(c_t,100)$.

The hierarchy $H$ consists of the interval $[A_1, A_{11})$ subdivided into intervals $[A_t, A_{t+1})$. As depicted in Figure~\ref{fig:gaussian}, each $[A_t, A_{t+1})$ is further subdivided into $[B_{t},D_{t})$ and its complement $[A_{t},B_{t}) \cup [D_{t},A_{t+1})$.
The key property is that half of the mass of $\mathrm{N}(0,100)$ lies in the corresponding interval $[B_t, D_t)$. For the above parameters: $A_t = c_t - 65$, $B_t = c_t - 6.6$, and $D_t = c_t + 6.6$.

The adversary $\A$ is described in Algorithm~\ref{alg:weak-downcoding-example}. It takes as input $\Y$, $k$, and a description of $H$. It looks at each group of generalized records $\hat{\Y}_t$ of the output. If the number of records in $\hat{\Y}_t$ is not $k$, then the whole group of records is copied to the output $\Z$ unchanged (i.e., no downcoding on these records).
If $\hat{\Y}_t$ has exactly $k$ records, then by $k$-anonymity these records are all identical copies of some $\y^t$. Some of $\y^t$'s entries may be aggregated to $[A_t, A_{t+1})$. If it's many more or many less than half the entries, then the whole group of records is copied to the output $\Z$ unchanged (i.e., no downcoding on these records). Otherwise, the $k$ records in $\hat{\Y}_t$ all get downcoded as described in the algorithm.

It follows from Example~\ref{ex:clustered-gaussians} that for $k=10$, the distribution described above, and $\Y$ produced by any minimal hierarchical $k$-anonymizer, $\A$ will downcode a constant fraction of the records in $\Y$ (with constant probability).

\begin{figure*}
  \setlength\abovecaptionskip{0.3em}
	\includegraphics[width=\textwidth]{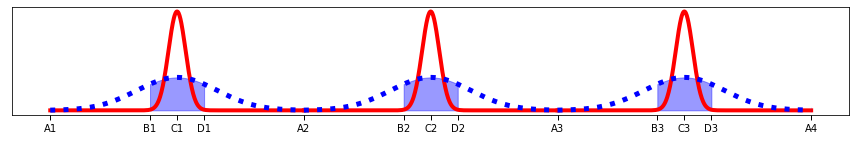}
	\caption{\label{fig:gaussian}The marginal distribution of each attribute for the example described in Section~\ref{sec:downcoding:example} (depicting 3 of 10 clusters, not to scale). A key property is that half of the mass of the $t^{\text{th}}$ blue dotted distribution lies in the interval $[B_t, D_t)$.}
\end{figure*}

\medskip
\begin{algorithm}
	\caption{Adversary $\A$ for the example in Section~\ref{sec:downcoding:example} (see also Fig.~\ref{fig:gaussian}).}
	\label{alg:weak-downcoding-example}
 \KwData{$\Y$, $k$}
 \KwResult{$\Z$}
 \For{cluster $t = 1,\dots,T$}{
	Let $\hat{\Y}_t$ be the records with an entry in $[A_t, A_{t+1})$\;
	\If{$|\hat{\Y}_t| \neq k$}{
		Copy every $\y \in \hat{\Y}_t$ into $\Z$\;
		continue\;
	}
	\BlankLine
	\tcc{$\hat{\Y}_t$ is $k$ exact copies of some $\y^t$}
	$\bg_t \gets \{d: y^t_d = [A_t, A_{t+1}]\}$\;
	$b_t \gets |\bg_t|$\;
	\eIf{$|b_t - D/2| > D/8$}{
		Write $k$ copies of $\y^t$ to $\Z$\;
	}{
		Write $k-1$ copies of $[B_t,D_t]$ to $\Z$\;
		Write $\z^t$ to $\Z$, where
		\vspace{-0.5em}
		\begin{equation*}
			z_d^t =
			\begin{cases}
			[B_t,D_t) & d \not\in\bg_t \\
			[A_t,B_t) \cup [D_i,A_{i+1}) & d \in \bg_t
			\end{cases}
		\end{equation*}\vspace{-1em}}}
\end{algorithm}

\ifUSENIX
\else

\subsection{Proof of Theorem~\ref{thm:weaker-downcoding}}

\begin{figure}[ht]
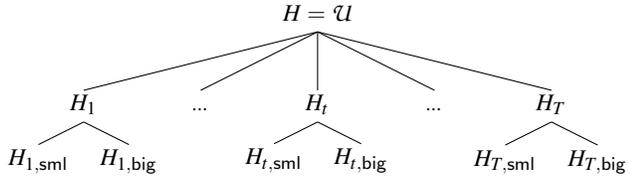

{\small
	\qtreecentertrue
	\Tree[
	.$H=\Att$
	    [.$H_1$
	        {$H_{1,\sm}$}
	        {$H_{1,\bg}$}
	    ]
			{\tiny \bf \ldots}
			[.$H_t$
	        {$H_{t,\sm}$}
	        {$H_{t,\bg}$}
	    ]
			{\tiny \bf \ldots}
	    [.$H_T$
	        {$H_{T,\sm}$}
	        {$H_{T,\bg}$}
	    ]
	]
}
\caption{\label{hierarchy:natural}The generalization hierarchy used in the proof of Theorem~\ref{thm:weaker-downcoding}. The attribute domain is an interval in $\bbR$, as are each $H_t$ and $H_{t,\sm}$. Each set $H_{t,\bg} = H_t \setminus H_{t,\sm}$ is the union of two intervals}
\end{figure}

\begin{claim}\label{claim:minimality-clustered}
	Let $H$ be a hierarchy with $T$ nodes at the second level: $H_1,\dots, H_t$ (as in Figure~\ref{hierarchy:natural}).
	Let $\X\in \Att^D$ be a dataset, $M$ be a minimal hierarchical $k$-anonymizer, and $\Y \gets M(\X, H)$.
 	For $t\in[1,T]$, let $\X_{t} = \X \cap H_t^D$ and let  $\Y_{t}$ be the records in $\Y$ corresponding to the records in $\X_{t}$.
	If $\X = \cup_{t} \X_t$, then for all but at most one $t\in \{t:|\X_t|=k\}$,\ $\y \subset H_t^D$ for all $\y \in \Y_t$.
\end{claim}

Note that as defined, the generalized records in $\Y_t$ are not necessarily contained in $H^D_t$. The claim says that if $\X$ consists of data in the $T$ clusters $\X_1,\dots,\X_T$, then the records in $\Y_t$ will be contained in $H^D_t$ for almost all clusters of size exactly $k$.

\begin{proof}[Proof of Claim~\ref{claim:minimality-clustered}]

	First, we show that for any $\y = (y_1,\allowbreak\dots,\allowbreak y_D)$, a single coordinate of $\y$ is generalized to $H = \Att$ if and only if \emph{every} coordinate in $\y$ is generalized to $H$. Namely, if $y_d = H$ for some $d$, then $\y = H^D$.

	Suppose for contradiction that there exists $\y = (y_1,\dots, y_D)$ corresponding to $\x\in \X$ such that $y_1 = H$ but $y_2 \subset H_t$ for some $t$.
	By the assumption on $\X$, there exists $t'$ such that $\y \in \Y_{t'}$. Because $y_2 \subset H_t$, $t'=t$ and hence $\y \in \Y_t$.
	By $k$-anonymity, there are at least $k-1$ additional records $\y'\in \Y$ such that $\y' = \y$. Repeating the previous argument, $\y' \in \Y_t$.

	Let $\y^* = (H_t, y_2,\dots, y_D)\subsetneq \y$.
	Construct $\Y^*$ by replacing all copies of $\y$ in $\Y$ with $\y^*$. It is immediate that $\Y^*$ is $k$-anonymous and respects the hierarchy. By the assumption that $y_1 = H$, $\Y^*$ strictly refines $\Y$. Additionally, $\Y^*$ generalizes $\X$, because all altered rows were in $\Y_t$.
	This contradicts the minimality of the $k$-anonymizer $M$. Therefore we have proved that if $y_d = H$ for some $d$, then $y_{d} = H$ for all $d$.

	Next, we show that for all but at most one $t\in \{t:|\X_t| =k\}$, there exists $\y\in\Y_t$ such that $\y\subset H^D_t$.
	By the preceding argument, it suffices to show that $\y \neq H^D$. Suppose for contradiction there exists  $t\neq t'$ such that for all $\y \in \Y_t \cup \Y_{t'}$, $\y = H^D$.
	Construct $\Y'$ by replacing each $\y\in \Y_{t'}$ with $H_{t'}^D\subsetneq \y$. It is easy to see that $\Y'$ respects the hierarchy, satisfies $k$-anonymity, generalizes $\X$, and strictly refines $\Y$. This contradicts the minimality of the $k$-anonymizer $M$.

	To complete the proof, let $t\in \{t:|\X_t| = k\}$ and suppose there exists $\y \in \Y_t$ such that $\y\subset H^D_t$. By $k$-anonymity, there must be at least $k-1$ distinct $\y' = \y \subset H^D_t$. By assumption on $\X$, each such $\y'$ must be an element of $\Y_t$.
	Because $|\Y_t| = |\X_t| = k$, every element of $\Y_t$ is equal to $\y\subset H^D_t$.
\end{proof}

\begin{proof}[Proof of Theorem~\ref{thm:weaker-downcoding}]
	\textbf{Data distribution}\quad
	Records $\x\sim U$ are noisy versions of one of $T = N/k$ cluster centers $\c_t \in \bbR^D$. Each coordinate $x_d$ of $\x$ is $c_{t,d}$ masked with i.i.d.\ noise with variance $\sigma^2$. The variance is usually \emph{small}, but is \emph{large} with probability $1/k$ (variances $\sigma^2_\sm \ll \sigma^2_\bg$).
	The generalization hierarchy $H$ is shown in Figure~\ref{hierarchy:natural}. $H$ divides the attribute domain $\Att$ into $T$ components $H_t$, each of which is further divided into {small values} $H_{t,\sm}$ and {large values} $H_{t,\bg}$.

	We set the parameters so that w.h.p.\ all of the following hold.
	First, the data is clustered: $\Pr_{\x}[\exists t,\ \x \in H_t^D] > 1-\negl(N)$.
	Second, every coordinate $x_d$ of a small-noise (variance $\sigma^2_\sm$) record is small: $x_d \in H_{t,\sm}$. Third, the coordinates of large-noise (variance $\sigma^2_\bg$) records are large or small ($x_d \in H_{t,\bg}$ or $x_d \in H_{t,\sm}$, respectively) with probability $1/2$ independent of all other coordinates . In particular, if $\x$ is generated using large noise then $\x \not\in H_{t,\sm}^D$ with high probability.
	An example of a distribution $U$ and hierarchy $H$ satisfying the above is given in Example~\ref{ex:clustered-gaussians}. In that example, the cluster centers $\c_t$ are masked with i.i.d.\ Gaussian noise.

\paragraph{The adversary}
	The adversary $\A$ takes as input $\Y$ and produces the output $\Z$ as follows. For $t \in [T]$, let $\hat{\Y}_t = \Y \cap H_t^D$.
	If $|\hat{\Y}_t| \neq k$, copy every $\y\in \hat{\Y}_t$ into the output $\Z$.
	Otherwise $|\hat{\Y}_t| = k$.
	By $k$-anonymity $\hat{\Y}_t$ consists of $k$ copies of a single generalized record $\y^t$.
	Let $\bg_t = \{d: y^t_d = H_t\}$ be the large coordinates of $\y^t$, and let $B_t = |\bg_t|$ be the number of large coordinates.
	If $\left| B_t - D/2\right| > D/8$, then $\A$ writes $k$ copies of $\y^t$ to the output $\Z$.
	Otherwise $\A$ writes $k-1$ copies of $H_{t,\sm}^D$ and one copy of $\z^t = (z^t_1,\ldots, z^t_D)$ to the output, where
	\begin{equation}
		\label{eq:z-weak}
		z_d^t =
		\begin{cases}
		H_{t,\bg} & d \in \bg_t \\
		H_{t,\sm} & d \not\in\bg_t
		\end{cases}
	\end{equation}

\paragraph{Analysis}
	It is immediate from the construction that $\Z \preceq \Y$. Moreover, it is easy to arrange the records in $\Z$ so that $\z_n\subset \y_n$ for all $n \in [N]$. By construction, $\z_n \subsetneq \y_n$ implies that $\z_n$ differs from $\z_n$ differs from $\y_n$ on at at least $B_t \ge 3D/8$ coordinates.

	To prove the theorem, it remains to show that w.h.p.\
	$\X \preceq \Z$ and $\Z\prec_{\Omega(N)} \Y$.
	Let $\X_\bg = \X \setminus \left(\bigcup_t\  H^D_{t,\sm}\right)$ consist of all the records $\x$ that have at least one large coordinate (i.e., $x_d \in H_{t,\bg}$ for some $t,d$).

	For all $t\in [T]$, let $\X_t = \X \cap H_t^D$ and let $\hat{\X}_t \subseteq \X_t$ be the records $\x \in \X$ that correspond to the records in $\hat{\Y}_t$.
	(Whereas $\X_t$ consists of all records that are in cluster $t$, $\hat{\X}_t$ consists of only those records that correspond to generalized records $\y \in \hat{\Y}_t$ that can be easily inferred to be in cluster $t$ based on $\Y$.)
	A cluster $t$ is \emph{$\X$-good} if $|\X_t| = k$ and $|\X_t \cap \X_\bg| =1$. A cluster $t$ is \emph{$\hat{\Y}$-good} if $\hat{\Y}_t =k$ and $|B_t - D/2|\le D/8$.

	It suffices to show that:
	\begin{itemize}[leftmargin=*]
		\item $\Omega(N)$ clusters $t$ are $\X$-good.
		\item All but at most one $\X$-good clusters are $\hat{\Y}$-good.
		\item For all $\hat{\Y}$-good clusters $t$, $\hat{\X}_t \cap \X_\bg = \{\x^t\}$ and $\x^t \subset \z^t$.
	\end{itemize}
	\noindent
	Note that if $t$ is both $\X$-good and $\hat{\Y}$-good, then $\hat{\X}_t = \X_t$. But there may be $t$ that are $\hat{\Y}$-good but not $\X$-good.

	\paragraph{Many clusters are $\X$-good} We lower bound $\Pr[t\mbox{ $\X$-good}]$ by a constant and then apply McDiarmid's Inequality.
	$$\Pr[t\mbox{ $\X$-good}] = \Pr[|\X_t| = k] \cdot \Pr\bigl[|\X_t \cap \X_\bg| = 1 \bigmid |\X_t| = k\bigr].$$
	$|\X_t|$ is distributed according to $\Binom(N,k/N)$, which approaches $\Pois(k)$ as $N$ grows. Using the fact that $k! \le (k/e)^k e\sqrt{k}$ we get:
		$\Pr[|\X_t| = k] \approx (k^k e^{-k})/k! \ge 1/(e\sqrt{k}) = \Omega(1).$
	$\Pr[\x \in \X_\bg] = \frac{1}{k} \pm \negl(N)$.
	The events $\x \in \X_t$ and $\x \in \X_\bg$ are independent.
		Therefore $$\Pr\bigl[|\X_t \cap \X_\bg| = 1 \bigmid |\X_t| = k\bigr] = \left(1-1/k\right)^{k-1} \pm \negl(N) > 1/e.$$
	Combining the above, $\Pr[t\mbox{ $\X$-good}] =\Omega(1)$. Quantitatively, for $k\le 15$, $\Pr[t\mbox{$\X$-good}] \gtrsim 1/(e^2\sqrt{k}) > 1/30$.

	Let $g(\X)$ be the number of $\X$-good values of $t$. By the above, $\E[g(\X)] = \Omega(N)$. Changing a single record $\x$ can change the value of $g$ by at most 2. Applying McDiarmid's Inequality,
	$$\Pr\left[g(\X) < \frac{\E(g(\X))}{2}\right] \le \exp\left(-\frac{2\left(\frac{\E(g(\X))}{2}\right)^2}{4N}\right) < \negl(N).$$
	Thus there are $\Omega(N)$ $\X$-good values of $t$ with high probability.

	\paragraph{Most $\X$-good clusters are $\hat{\Y}$-good}
	Cluster $t$ is $\hat{\Y}$-good if $\hat{\Y}_t =k$ and $|B_t - D/2|\le D/8$.
	First we show that for all but one $\X$-good $t$, $|\hat{\Y}_t| =k$. Let $\Y_{t}\supseteq \hat{\Y}_t$ be the records in $\Y$ corresponding to the records in $\X_{t}$. (Whereas $\Y_t$ consists of all records that correspond to $\X_t$, $\hat{\Y}_t$ consists of only those records whose membership in $\Y_t$ can be easily inferred from $\Y$.)
	Observe that $|\X_t| = |\Y_t| \ge |\hat{\Y}_t|$.
	By construction, for all $\x\in\X$ there exists $t$ such that $\x \in \X_t$ with high probability (i.e., $\X =  \cup_t \X_t$).
	By Claim~\ref{claim:minimality-clustered}, for all but at most one $\X$-good $t$ and every $\y \in \Y_t$, $\y \subset H_t^D$. Thus $|\hat{\Y}_t| =\Y_t = k$.

	Finally we show that for all $\X$-good $t$ as guaranteed by Claim~\ref{claim:minimality-clustered}, $B_t \in (3D/8, 5D/8)$ with high probability.
	$\hat{\Y}_t$ consists of $k$ copies of the same generalized record $(y_1,\dots, y_d)$.
	Since $M$ is hierarchical, $y_d\in \{H_t,H_{t,\sm},H_{t, \bg}\}$.
	By the $\X$-goodness of $t$, $\X_t$ contains $1$ large-noise record $\x_\bg$ and $k-1$ small-noise records $\x'$
	By correctness of the $k$-anonymizer $M$, $x_{\bg,d} \in H_{t,\bg}$ $\implies y_d \supseteq H_{t,\bg}$.
	Minimality implies the converse: $y_d \supseteq H_{t,\bg}$ $\implies x_{\bg,d} \in H_{t,\bg}$.
	With high probability, $x'_d \in H_{t,\sm}$ $\implies y_d\supseteq H_{t,\sm}$.
	Putting it all together,
	$$ x_{\bg,t} \in H_{t,\bg}\iff y_d = H_t \iff d\in \bg_t.$$
	By construction of the data distribution $U$, $\Pr[d \in \bg_t] = 1/2$ independently for each $d \in [D]$. Applying Chernoff again,
	$\Pr[|B_t - D/2| \ge D/8]< 2e^{-\Omega(D)}<\negl(N)$.

	\paragraph{Analyzing $\hat{\Y}$-good clusters}
	By construction, $\bg_t = \{d : \exists \x \in \hat{\X}_t \cap \X_\bg \text{ st } x_d \in H_{t,\bg}\}$.
	Because $|\bg_t| > 0$, $|\hat{\X}_t \cap \X_\bg| > 1$.
	A simple Chernoff-then-union-bound argument shows that the probability that there exist distinct records  $\x,\x' \in \X_\bg$ such that $|\bg_t| = |\{d: x_d\in H_{t,\bg} \lor x'_d\in H_{t,\bg}\}| < 5D/8$ is negligible.
	Hence $\hat{\X}_t\cap\X_\bg$ is a singleton $\{\x^t\}$ with high probability.
	$\x^t \subseteq \z^t$ follows immediately from the construction.
\end{proof}

\begin{example}\label{ex:clustered-gaussians}
	The following distribution $U$ and hierarchy $H$ suffice for the proof of Theorem~\ref{thm:weaker-downcoding}.
	The distribution $U$ is defined by $T = N/k$ cluster centers $c_1,\dots,c_{T} \in \bbR$ and standard deviations $\sigma_\sm,\sigma_\bg \in \bbR$.
	A record $\x \in \bbR^D$ is sampled as follows. Sample a cluster center $t\gets[T]$ uniformly at random. Sample $\sz \gets \{\sm, \bg\}$ with $\Pr[\sz = \bg] = 1/k$.
	Sample noise $\e \gets \Norm(0,\sigma^2_\sz \vec{I})$.
	Output $\x = \c_t + \e$, where $\c_t = (c_t, \dots, c_t) \in \bbR^D$.

	The hierarchy $H$ consists of intervals $H_t = [c_t - \Delta,c_t + \Delta]$ centered at the cluster centers $c_t$, for some $\Delta$. The hierarchy further subdivides each $H_t$ into a smaller interval $H_{t,\sm} = [c_t - \tau, c_t + \tau)$, for some $\tau<\Delta$, and the complement $H_{t,\bg} = H_t \setminus H_{t,\sm}$.

	To suffice for our proof, we require that
	with high probability over $\x \sim U$ there exists $t\in [T]$ such that:
	(a) $\x \in H_t^D$; (b) if $\sz = \sm$, then $\x \in H_{t,\sm}^D$;
	(c) if $\sz = \bg$, then each coordinate $x_d$ is in $H_{t,\bg}$ with probability $1/2$ independent of all other coordinates $x_{d'}$.
	Many instantiations of the parameters would work, such as:
 	$\sigma_\sm = 1$, $\tau =\log N$, $\sigma_\bg = \zeta\log N$, $\Delta = \zeta \log^2 N$, and $c_t = 2t\Delta$, where $\zeta = \frac{1}{\sqrt{2}\cdot \mathrm{erf}^{-1}(1/2)} \approx 1.48$ and $\mathrm{erf}$ is the Gaussian error function.
\end{example}

\ifUSENIX
	\subsection{Proof of Theorem~\ref{thm:full-downcoding}}
	\label{app:proof-full-downcoding}
\else
	\subsection{Proof of Theorem~\ref{thm:full-downcoding}}
\fi

\begin{figure}
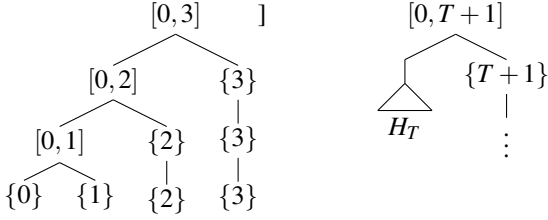

	\begin{tabular}{cc}
\qtreecentertrue
\Tree[
	.$[0,3]$
			[.$[0,2]$
	        [.$[0,1]$
						{$\{0\}$}
						{$\{1\}$}
					]
					[.$\{2\}$
						{$\{2\}$}
					]
	    ]
			[.$\{3\}$
				[.$\{3\}$
					{$\{3\}$}
				]
			]
	]
]
&
\quad
\quad
\quad
\qroofy=60
\qroofx=60
\Tree [.$[0,T+1]$ [ \qroof{$H_T$}. ] [.$\{T+1\}$
		[.$\vdots$
		]
] ]
\end{tabular}
\caption{\label{fig:main-thm-hierarchy}The generalization hierarchy used in the proof of Theorem~\ref{thm:full-downcoding}. On the left is the hierarchy $H_3$ for the domain $\Att = [03]$. On the right is a recursive construction of $H_{T+1}$ for domain $\Att = [0,T+1]$ from the hierarchy $H_T$.}
\end{figure}

A $k$-anonymizer $M:\X\mapsto\Y$ groups records $\x \in \X$ into equivalence classes such that if $\x$ and $\x'$ are in the same class, then $\Y(\x) = \Y(\x')$.
In general, $M$ may have a lot of freedom to group the $\x$'s the equivalence classes and also to choose the $\y$'s that generalize each equivalence class.

Claim~\ref{claim:minimal-M-structural} states that if $M$ is minimal and generalizes using hierarchy like in Figure~\ref{fig:main-thm-hierarchy}, then it has much less freedom. Namely, $\Y$ is fully determined by the choice of equivalence classes (with probability at least $1-\alpha$ over the dataset $\X$). $M$ can group the $\x$'s together, but then has no control over the resulting $\y$'s.

Claim~\ref{claim:minimal-M-structural} and its proof are meant to be read in the context of the proof of Theorem~\ref{thm:full-downcoding} and freely uses its notation.

\begin{proof}[Proof of Theorem~\ref{thm:full-downcoding}]
Let $T = \lceil N^2/\alpha \rceil$ and $\Att = [0,T]$ be the attribute domain. Records $\x \in \Att^D$ are sampled according to the distribution $U$ as follows. First sample $t(\x) \gets [T]$ uniformly at random. Then sample each coordinate $x_d$ of $\x$ i.i.d.\ with $\Pr[x_d = t(\x)] = 1/2k$ and $x_d =0$ otherwise.
In other words, $\x\in \{0,t(\x)\}^D$ consists of $D$ independent samples from $t(\x) \cdot \Bern(1/2k)$.

All the $t(\x)$ will be distinct except with probability at most ${N \choose 2} \frac{1}{T}< \alpha/2$. If all $t(\x)$ are distinct, we say $\X$ is \emph{collision-free.}
The remainder of the proof shows that the adversary succeeds with high probability conditioned on $\X$ collision-free.

Figure~\ref{fig:main-thm-hierarchy} defines the generalization hierarchy. It consists of intervals $[0,t]$ and singletons $\{t\}$ for $t \in [T]$.

Claim~\ref{claim:minimal-M-structural} states that the output $\Y\gets M(\X,H)$ of a minimal hierarchical $k$-anonymizer must take a restricted form.
For $\y\in\Y$, let $\X_\y = \{\x \in \X : \Y(\x) = \y\}$ be the records in $\X$ that correspond to a copy of $\y \in \Y$.
The claim states that
$$\Pr\biggl[\max_\y|\X_\y| < 2k \Bigmid \X \mbox{ collision-free}\biggr] > 1-\negl(N).$$
Moreover, if $\X$ is collision-free then for all $\y \in \Y$ and $d \in [D]$:
\begin{equation}\label{eqn:1n23l1}
	y_d = [0,\max_{\x\in\X_\y}x_d].
\end{equation}

Let $\A$ be deterministic adversary that on input $\Y$ does the following. For $t$, pick $\y^t = (y^t_1,\dots,y^t_D) \in \Y$ such that $\exists d\in[D]$, $y^t_d = [0,t]$. Let $\y^t = \bot$ if no such $d$ exists.
By the \eqref{eqn:1n23l1}, all $\y$ satisfying the above are identical.
If $\y^t\neq \bot$, we define the following subsets of $[D]$:
\begin{align*}
&D^t(\Y) =\{d : y^t_d = [0,t]\} \\
&D^{>t}(\Y) =\{d : y^t_d = [0,t'] \mbox{ for } t'>t\} \\
&D^{<t}(\Y) =\{d : y^t_d = [0,t'] \mbox{ for } t'<t\}.
\end{align*}

If $\y^t \neq \bot$, $\A$ writes $\z^t = (z^t_1,\ldots,z^t_D)$ to the output $\Z$, where
\begin{equation}
		\label{eq:z-full}
	z_d^t =
	\begin{cases}
	0 & d \in D^{<t}(\Y) \\
	t & d \in D^{t}(\Y) \\
	[0,t] & d \in D^{>t}(\Y)
	\end{cases}
\end{equation}

Let $T_\X = \{t : \y^t \neq \bot\}$.
$|\Z| = |T_\X|$, and it is easy to see that $\Pr[|T_\X| = N \mid \X\mbox{ collision-free}] > 1-\negl(N)$.
Hence if $\X$ is collision-free, then $\Z \prec_N \Y$ by construction. In this case, we assume without loss of generality that the rows in $\Z$ are ordered in a way that $\z_n \subseteq \y_n$.
It follows immediately from the construction that $\z_n \subsetneq_D \y_n$.

If $\X$ is collision-free, then for every $t\in T_\X$ there is a unique $\x^t = (x^t_1,\ldots,x^t_D) \in \X$ such that $t(\x^t) = t$.
By \eqref{eqn:1n23l1} and the fact that $\x^t \in \{0,t\}^D$, $\x^t \subset \z^t$. Hence if $\X$ is collision-free, then $\X \preceq \Z$ with high probability, proving the first part of the theorem.
\end{proof}

The following claims is meant to be read in the context of the proof of Theorem~\ref{thm:full-PSO} and freely uses notation therefrom.

\begin{claim}\label{claim:minimal-M-structural}
For $\y \in \Y$, let $\X_\y = \{\x \in \X : \Y(\x) = \y\}$ be the records in $\X$ that correspond to a copy of $\y \in \Y$.
If $\X$ is collision-free, then for all $\y\in\Y$ and $d\in[D]$:
\begin{equation*}
	y_d = [0,\max_{\x \in \X_\y} x_d].
\end{equation*}
Moreover, $\Pr[\max_\y |\X_\y|<2k \mid \X\mbox{ collision-free}] > 1-\negl(N)$.
\end{claim}
\begin{proof}
	Both parts of the claim rely on the minimality of $M$.

	Recall that $\x \in \{0,t(\x)\}^D$.
	Let $T^*_d(\X_\y) = \{x_d : \x \in \X_\y\}\subset \cup_{\x \in \X_\y} \{0,t(\x)\}$ be the set of all values in the $d$th column of $\X_\y$.
	$\X$ collision-free implies that either $0\in T^*_d(\X_\y)$ or $|T^*_d(\X_\y)|\ge 2$ (probably both).
	Because $M$ is correct and hierarchical, $T^*_d(\X_\y)\subset y_d \in H$.
	Hence, by construction of $H$, $y_d = [0,t_d]$ for some $t_d \in [0,T]$.
	Let $t_d^* = \max_{\x \in \X_\y} x_d$.
	Correctness requires $[0,t_d^*] \subset [0,t_d]$.
	Moreover, replacing $\y = [0,t_d]$ with $[0,t_d^*]$ would yield a $k$-anonymous, hierarchy-respecting refinement of $\Y$. By minimality of $M$, $[0,t_d] \subset [0,t_d^*]$. Hence, $y_d = [0,t_d^*]$.

	It remains to prove the bound on $\max_\y |\X_\y|$.
	Let $\X_0$ and $\X_1$ be an arbitrary partition of $\X_\y$.
	For $b \in \{0,1\}$, define $\y_b'\subset \y$ as:
	$$\y_b' = (y_{b,1}', \dots, y_{b,D}') = ([0,\max_{\x \in \X_b}x_1], \dots, [0,\max_{\x \in \X_b}x_D])$$

	$\Pr[\exists b,~\y_b' \subsetneq \y \mid \X\mbox{ collision-free}] > 1-\negl(N)$.
	To see why, observe that $\y_0' = \y = \y_1'$ implies that for every coordinate $d$, $\max_{\x \in \X_0}(x_d)  = \max_{\x \in \X_\y}(x_d)= \max_{\x \in \X_1}(x_d)$.
	If $\X$ is collision-free, this implies that for all $d$, $\max_{\x \in \X_\y}(x_d) = 0$.
	This occurs with probability $1-2^{-D} = 1-\negl(N)$ (even conditioned on collision-free).

	Consider $\Y'$ constructed by replacing every instance of $\y$ in $\Y$ with $\y'_0$ or $\y'_1$, using $|\X_0|$ and $|\X_1|$ copies respectively. By construction, $\Y'$ correctly generalizes $\X$ and respects the hierarchy $H$. By the preceding argument, $\Y'$ strictly refines $\Y$ with high probability.
	Thus, by minimality of $M$, $\Y'$ cannot be $k$-anonymous. This means that for every partition $\X_0,\X_1$, one of $|\X_b|\le k-1$. Therefore, $|\X_\y| < 2k$.

\end{proof}

The following claim is used to prove Theorem~\ref{thm:full-PSO}. It is meant to be read in the context of Theorem~\ref{thm:full-downcoding} and freely uses notation therefrom.
\begin{claim}
	\label{claim:well-spread}
	Let $k\ge 2$, $D = \omega(\log N)$, $U$, $\X$, $\Y$, and $D^t(\Y)$ as defined in the proof of Theorem~\ref{thm:full-downcoding}. Let $T' = \{t : \z^t \in \Z\}$.
	\begin{align}
	\Pr\biggl[\forall t \in T' : |D^t(\Y)| \ge \frac{D}{4ke} \Bigmid \X\mbox{ coll-free}\biggr]
	&> 1-\negl(N) \label{eq:claim:well-spread}
	\end{align}
	Equation \eqref{eq:claim:well-spread} also holds for $k = N$, $D = \omega(N\log N)$.
\end{claim}

\begin{proof}[Proof of Claim~\ref{claim:well-spread}]
	The proof is an application of Chernoff and union bounds.
	We rewrite $D^t(\Y)$ as $\{d : \exists n,~\Y_{n,d} = [0,t]\}$.

	For $\y \in \Y$, let $\X_\y$ contain the records $\x$ that correspond to a copy of $\y$.
	Consider $\x^*\in \X_\y$, and let $t^*  = t(\x^*)$.
	We call $d$ SUPER if $(x^*_{d} \neq 0)$ and $(x'_d = 0 \text{ for all } \x' \in \X_\y\setminus\{\x^*\})$.
	By Claim~\ref{claim:minimal-M-structural}, if $d$ is SUPER then $d \in D^t(\Y)$.
	We will lower bound the number of SUPER $d$.

	For an index set $I\subseteq [N]$, let $\X_I = \{\x_n\}_{n\in I}$.
	By Claim~\ref{claim:minimal-M-structural},
	$$\Pr[\exists I \text{ st } (|I|<2k) \land (\X_\y = \X_I) ~\mid~ \X\mbox{ coll-free}] > 1-\negl(N).$$
	Observe that if $\X_\y = \X_I$, then $\x^* \in \X_I$ and $|I| \ge k$ (by $k$-anonymity).

	We call $d$ GOOD with respect to $\X_I$ if there is a unique $\x \in \X_I$ such that $x_d \neq 0$. Let $D_I = \{d \text{ GOOD wrt } \X_I\}$.
	Observe that if $\X_\y = \X_I$ and $d$ is SUPER, then $d$ is GOOD with respect to $\X_I$. Therefore
	$$|D^t| \ge \min_{I : k\le|I|<2k, \x^* \in I} |D_I|$$
	except with at most negligible probability (conditioned on $\X$ collision-free).

	For fixed $I$, $\E[|D_I|] = D(1/2k)(1-1/2k)^{|I|} > D/2ke$ where the last inequality follows from  $|I|<2k$.
	By a Chernoff bound:
	$	\Pr[|D_I| \le D/4ke] \le e^{-D/32ke}.$
	By the union bound:
	\begin{align*}
	\Pr\left[\min_{I : k\le|I|<2k, \x^* \in I} |D_I| \le D/4ke\right]
	&\le e^{-D/32ke}\sum_{|I| = k}^{2k-2} {N\choose |I|}
\end{align*}
For $k = O(1)$, $\sum_{|I| = k}^{2k-2} {N\choose |I|} < N^{2k-2}$.
Putting it all together with a final union bound:
\begin{equation*}
	\Pr\biggl[\exists t,~|D^t| \le D/4ke \Bigmid \X\mbox{ collision-free}\biggr]
	< N^{2k-1}e^{-D/32ke}.
\end{equation*}
For $D = \omega(\log N)$, this probability is negligible.
For $k=N$ and $D = \omega(N\log N)$, the set $\{I : k \le |I| < 2k\}$ is a singleton, doing away with the need for a union bound. In this case the upper bound is $Ne^{-D/32Ne} < \negl(N)$.
\end{proof}

\fi

\section{Predicate singling-out attacks on syntactic privacy techniques}
\label{sec:pso}

Our downcoding attacks yield powerful predicate singling-out (PSO) attacks against minimal hierarchical $k$-anonymous mechanisms.
PSO attacks were recently proposed as a way to demonstrate that a privacy mechanism fails to legally anonymize under Europe's General Data Protection Regulation \cite{altman2020hybrid, CohenNissimSinglingOut}. Our new attacks undermine the use $k$-anonymity and other QI-deidentification techniques for GDPR compliance, challenging prevailing European guidance on anonymization \cite{wp-anonymisation}.

In this section, we recall the prior work on PSO attacks and define a generalization called \emph{compound PSO attacks}. We prove that minimal hierarchical $k$-anonymizers enable compound PSO attacks.

\subsection{Background on PSO attacks}

Predicate singling-out attacks were recently introduced by Cohen and Nissim in the context of data anonymization under Europe's General Data Protection Regulation (GDPR) \cite{CohenNissimSinglingOut}.
They were proposed as a mathematical test to show that a privacy mechanism fails to legally anonymize data under GDPR \cite{altman2020hybrid, CohenNissimSinglingOut}. A mechanism $M$ legally anonymizes under GDPR if it suffices to transform regulated \emph{personal data} into unregulated \emph{anonymous data}. That is, if $M(\X)$ is free from GDPR regulation regardless of what $\X$ is. If a mechanism enables PSO attacks, then it does not legally anonymize under GDPR \cite{altman2020hybrid}.

Informally, $M$ enables PSO attacks if given $M(\X)$, an adversary is able to learn an extremely specific description $\psi$ of a single record in $\X$. Because $\psi$ is so specific, it not only distinguishes the victim in the dataset $\X$, but likely also in the greater population.
Hence PSO attacks can be a stepping stone to more blatant attacks.

Formally, we consider a dataset $\X = (\x_1, \ldots, \x_n)$ sampled i.i.d.\ from distribution $U$ over universe $\Att^D$.
The PSO adversary $A$ is a non-uniform probabilistic Turing machine which takes as input $M(\X)$ and produces as output a \emph{predicate} $\psi:\Att^D\to\zo$.
$\psi$ \emph{isolates} a record in a dataset $\X$ if there exists a unique $\x \in \X$ such that $\psi(\x) = 1$. Equivalently, if $\psi(\X) = \sum_{\x \in \X} \psi(\x)/n = 1/n$.
The strength of a PSO attack is related to the \emph{weight} of the predicate $\psi$ output by $A$: $\psi(U) \triangleq \E(\psi(\x))$ for $\x\sim U$. We simplify the definitions from \cite{CohenNissimSinglingOut} to their strongest setting: where $\psi(U) < \negl(n)$.

To perform a PSO attack, $\adv$ outputs a single negligible-weight predicate $\psi$ that isolates a record $\x\in\X$ with non-negligible probability.

\begin{definition}[Predicate singling-out attacks (simplified) \cite{CohenNissimSinglingOut}]
	\label{def:PSO}
	$M$ \emph{enables predicate singling-out (PSO) attacks} if there exists $U$, $\adv$, and $\beta(n)$ non-negligible such that
	$$\Pr_{\substack{\X\gets U^n \\ \psi\gets \adv(M(\X))}}[\psi(\X) = 1/n \land \psi(U)< \negl(n)] \ge \beta(n).$$
\end{definition}

Cohen and Nissim give a simple PSO attack against a large class of $k$-anonymizers which they call \emph{bounded}. A $k$-anonymizer is bounded if there is some maximum $k_{\mathrm{max}}$ such that for all $\X$, the effective anonymity of every row of $\Y$ is at most $k_{\mathrm{max}}$.
The attacker outputs $L = O(N)$ disjoint negligible-weight predicates $\psi$.
If $M$ is bounded, each $\psi$ isolates a row in $\X$ with probability about $\eta/e\gg 0$ independently, where  $\eta \in [0,1]$ is a parameter that depends on $M$ and $U$.

\subsection{Compound predicate singling-out attacks}

PSO attacks can be unsatisfying.
For example, the attack from \cite{CohenNissimSinglingOut} outputs $L$ predicates and at best about $L/e$ manage to actually isolate a record in the dataset $\X$. Moreover, which predicates isolate and which don't is impossible for the attacker to know without additional information.
So even though there exists many isolated records with high probability, the attacker doesn't know which ones or how many.
In contrast, consider an attacker that outputs $L = N$ predicates, each of which isolates a distinct record in $\X$. It is obvious the new attacker is stronger, but in a way that isn't captured by the definition of predicate singling-out.

We define a generalization of PSO attacks called \emph{compound} PSO attacks.
Whereas PSO attacks only require that a record is isolated with non-negligible probability, compound PSO attacks require many records to be isolated often.

To perform a compound PSO attack, $\adv$ outputs multiple negligible-weight predicates $\Psi = \{\psi_1,\dots,\psi_L\}$ each of which isolates a distinct record $\x \in \X$ with probability at least $1-\alpha$. The strength of the attack is measured by $L$ and $\alpha$, with $L\to n$ and $\alpha\to 0$ reflecting stronger attacks.
Vanilla PSO attacks correspond to the setting $L=1$ and $\alpha = 1-\beta$.

\begin{definition}[$(\alpha,L)$-compound-PSO attacks]
	\label{def:compound-PSO}
	$M$ \emph{enables $(\alpha,L)$-compound predicate singling-out attacks} if there exists $U$, $\adv$ such that
	\begin{equation*}
	\Pr\biggl[\forall \psi,\psi'\in \Psi:
	\begin{array}{c}
		\psi(\X) =1/n
		\land \psi(U) < \negl(N) \\
		\land (\psi\wedge \psi')(U) = 0
		\land |\Psi| \ge L
		\end{array}
		\biggr] \ge 1-\alpha(N)
\end{equation*}
	in the probability experiment $\X\sim U^N$, $\Psi\gets \A(M(\X))$.
\end{definition}

In the language of compound attacks, the prior work gives an $(1-O(e^{-L}), L)$-compound-PSO attack against bounded $k$-anonymizers
for $L < cN$ and some $c>0$.

Our compound PSO attacks are much stronger. Theorem~\ref{thm:weaker-PSO} gives a $(\negl(N), \Omega(N))$-compound-PSO attack, and Theorem~\ref{thm:full-PSO} gives a $(\poly(1/N), N)$-compound-PSO-attack.
In both attacks, the adversary fails only if the dataset $\X$ is atypical in some way. If the dataset is typical, the compound PSO attack always succeeds regardless of what the mechanism $M$ does.
The tradeoff is that our new attacks only work on minimal hierarchical $k$-anonymizers (instead of all bounded $k$-anonymizers) and with more structured data distributions $U$ (instead of any $U$ with moderate min-entropy).

\begin{theorem}\label{thm:weaker-PSO}
For all $k \ge 2$, $D = \omega(\log N)$, there exist a distribution $U$ over $\mathbb{R}^D$, a generalization hierarchy $H$, such that all minimal hierarchical $k$-anonymizers $M$ enable $(\negl(N),\Omega(N))$-compound-PSO attacks.
\end{theorem}

\begin{theorem}\label{thm:full-PSO}
	For all constants $k \ge 2$, $\alpha > 0$, $D = \omega(\log N)$, and $T = \lceil N^2/\alpha \rceil$, there exists a distribution $U$ over $\Att^D = [0,T]^D$, a generalization hierarchy $H$, such that all minimal hierarchical $k$-anonymizers $M$ enable $(\alpha,N)$-compound-PSO attacks.
	The attack also works for $k=N$ and $D = \omega(N \log N)$.
\end{theorem}

These theorems mirror Theorems~\ref{thm:weaker-downcoding} and \ref{thm:full-downcoding}, inheriting their advantages and disadvantages.
Proofs for both attacks follow the same general structure, using the corresponding downcoding attacks in non-black-box ways (Appendix~\ref{app:proofs}).
The key observation is that some of the downcoded records in the downcoding attacks immediately give the predicates needed to predicate single-out.

Algorithm~\ref{alg:PSO-example} illustrates the compound-PSO adversary for the example of clustered Gaussians described in Section~\ref{sec:downcoding:example}. Compare to the downcoding adversary in Algorithm~\ref{alg:weak-downcoding-example}. Instead of outputting a complete dataset $\Z$ (as in the downcoding attack), we simply output descriptions of certain records within $\Z$.
Namely, $\mathsf{matches}(\z):\x \mapsto \{0,1\}$ is the predicate that outputs $1$ if and only if $\x$ is consistent with $\z$ (i.e., $\x \subseteq \z)$.

\medskip
\begin{algorithm}
	\label{alg:PSO-example}
 \KwData{$\Y$, $k$}
 \KwResult{$\Psi$}
 \For{cluster $t = 1,\dots,T$}{
	Let $\hat{\Y}_t$ be the records with an entry in $[A_t, A_{t+1})$\;
	\If{$|\hat{\Y}_t| \neq k$}{
		continue\;
	}
	\BlankLine
	$\bg_t \gets \{d: y^t_d = [A_t, A_{t+1}]\}$\;
	$b_t \gets |\bg_t|$\;
	\eIf{$|b_t - D/2| > D/8$}{
		continue\;
	}{
		$\Psi\gets \Psi \cup \{\mathsf{matches}(\z^t)\}$, where
		\vspace{-0.5em}
		\begin{equation*}
			z_d^t =
			\begin{cases}
			[B_t,D_t) & d \not\in\bg_t \\
			[A_t,B_t) \cup [D_i,A_{i+1}) & d \in \bg_t
			\end{cases}
		\end{equation*}
		\vspace{-1em}}}
 \caption{Compound-PSO adversary for the example in Section~\ref{sec:downcoding:example} (compare with Alg.~\ref{alg:weak-downcoding-example}).}
\end{algorithm}

\ifUSENIX
\else

\begin{proof}[Proof outline]
Proofs for both compound PSO attacks follow the same general structure, using the corresponding downcoding attacks in non-black-box ways.
The compound PSO adversary $\A$ gets as input $\Y \gets M_H(\X)$.
It emulates the appropriate downcoding adversary, which produces an output $\Z$ such that $\X \preceq \Z \prec \Y$.

$\Z$ contains special generalized records $\z^t$ indexed by some $t \in [T]$ (equations~\eqref{eq:z-weak} and \eqref{eq:z-full}), and may also contain other records. $\A$ outputs $\Psi = \{\psi^t\}$ where $\psi^t:\x \mapsto \indic{\x \subseteq \z^t}$.

To complete the proof, one must show that the following hold with probability at least $1-\alpha(N)$:
\begin{itemize}
	\item $\psi(\X) =1/N$
	\item $(\psi\wedge \psi')(U) = 0$
	\item $\psi(U) < \negl(N)$
	\item $|\Psi| \ge L$
\end{itemize}
The first three are implied by the following:
\begin{itemize}
\item  $\forall t$, there exists a unique $\x^t \in \X$ such that $\x^t \subseteq \z^t$.
\item  $\forall t\neq t'$, $\z^t \cap \z^{t'} = \emptyset$.
\item  $\forall t$, $\x \sim \Att^D$, $\Pr_{\x}[\x \subset \z^t] < \negl(N)$.
\end{itemize}
For the downcoding attack from Theorem~\ref{thm:weaker-downcoding}, these properties are immediate. For the downcoding attack from Theorem~\ref{thm:full-downcoding}, the first two are immediate and the third follows from Claim~\ref{claim:well-spread}.

The requirements on $L$ and $\alpha$ follow from the parameters of the corresponding downcoding attacks. The downcoding attack for Theorem~\ref{thm:weaker-downcoding} yields $\Omega(N)$ records $\z^t$ with probability $1-\negl(N)$.
The downcoding attack for Theorem~\ref{thm:full-downcoding} yields $N$ records $\z^t$ with probability $1-\alpha(N)$.
\end{proof}

\fi

\section{Reidentifying EdX students using LinkedIn}
\label{sec:edx}

\newcommand{\Qall}{Q_{\mathsf{all}}}
\newcommand{\Qfriend}{Q_{\mathsf{acq}}}
\newcommand{\QfriendEdu}{Q_{\mathsf{acq+}}}
\newcommand{\QLI}{Q_{\mathsf{resume}}}
\newcommand{\Qclassmate}{Q_{\mathsf{posts}}}

\newcommand{\unambig}{\mathsf{u}}

\newcommand{\EdXraw}{\X_{\mathsf{ed,raw}}}
\newcommand{\EdX}{\X_{\mathsf{ed}}}
\newcommand{\EdXclean}{\X_{\mathsf{full\text{-}demo}}}

597,692 individuals registered for 17 online courses offered by Harvard and MIT through the EdX platform \cite{edX_first_year}.
We show that thousands of these students are potentially vulnerable to reidentification.
The EdX dataset represents an egregious failure of $k$-anonymity in practice and in a case where the dataset was ``properly deidentified'' by ``statistical experts'' in accordance with regulations, undermining one of the main arguments used to justify the continued use of QI-deidentification  \cite{ElEmamSystematic}.

EdX collected data about students' demographics, engagement with course content, and final course grade.
EdX sought to make the data public to enable outside research but considered it protected by the Family Educational Rights and Privacy Act (FERPA), a data privacy law restricting the disclosure of certain educational records \cite{edX-dataset}.
``To meet these privacy specifications, the HarvardX and
MITx research team (guided by the general counsel, for the
two institutions) opted for a $k$-anonymization framework'' \cite{edX-how-to}. A value of $k=5$ ``was chosen to allow legal sharing of the data'' in accordance with FERPA.
Ultimately, EdX published the 5-anonymized dataset with 476,532 students' records.

We show that thousands of these students are potentially vulnerable to reidentification. As a proof of concept, we reidentified 3 students out of 135 students for whom we searched for matching users on LinkedIn. Each of the reidentified users {failed to complete} at least one course in which they were enrolled, a private fact disclosed by the reidentification attack.

The limiting factor of this attack was not the privacy protection offered by $k$-anonymity itself, but the fact that many records in the raw dataset were missing demographic variables altogether. In order to boost the confidence of our attack, we restricted our attention to  \emph{unambiguously} unique records. To demonstrate the possibility of attribute disclosure, we further restricted our attention to students that had enrolled in, but failed to complete, a course on EdX.

\subsection{The Harvard-MIT EdX Dataset}

$\EdX$ has 476,532 rows, one per student.\footnote{%
  The dataset as published was such that each row represented a student-course pair, with a separate row for each course in which a student enrolled. Records corresponding to the same student shared a common UID. $\EdX$ as described above is the result of aggregating the information by UID.
  See the appendix for additional background on the EdX dataset.}
Each row contains the student's basic demographic information, and information about the student's activities and outcomes in each of 16 of the 17 EdX courses.

The demographics included self-reported level of education, gender, and year of birth, along with a country inferred from the student's IP address.
Many students chose not to report level of education, gender, and year of birth at all, so these columns are missing many entries.
For each course, $\EdX$ indicates whether the student enrolled in the course, their final grade, and whether they earned a certificate of completion. $\EdX$ also includes information about students' activities in courses including how many forum posts they made.

$\EdX$ was $5$-anonymized with respect to 17 overlapping quasi-identifiers separately: $Q_1,\dots,Q_{16}$, and $Q_*$ defined next. Recall that each quasi-identifier is a subset of attributes, not a single attribute (Def.~\ref{def:k-anon}).
\begin{itemize}[leftmargin=*]
  \item $Q_i =$ \{gender,  year of birth, country, enrolled in course $i$, number of forum posts in course $i$\}
  \item $Q_* = $ \{enrolled in course 1, \dots, enrolled in course 16\}.
\end{itemize}
Anonymization was done hierarchically. First, locations were globally coarsened to countries or continents. Then other attributes or whole records were suppressed as needed.

\subsection{Uniques in the EdX dataset}

Table~\ref{table:results-k-anon} summarizes the results of all analyses described in this section.
Let $\Qall = Q_* \cup Q_1 \cup \dots \cup Q_{16}$.
$\EdX$ is very far from 5-anonymous with respect to $\Qall$.
We find that 7.1\% of students (33,925 students) in $\EdX$ are unique with respect to $\Qall$ and 15.3\% have effective anonymity less than 5.

Despite EdX's goals, $\EdX$ was not even $5$-anonymous with respect to $Q_*$: 245 students were unique and 753 had effective anonymity less than 5!
We suspect this blunder is due to $k$-anonymity's fragility with respect to post-processing.
The raw data was first 5-anonymized with respect to $Q_*$ and afterwards with respect to $Q_1,\dots,Q_{16}$. Some rows in the dataset were deleted in the latter stage, ruining $5$-anonymity for $Q_*$.

We emphasize that the creators of the EdX dataset never intended or claimed to provide 5-anonymity with respect to $\Qall$. But they admit that each of the attributes in $\Qall$ is potentially public.
In our view, the union of quasi-identifiers should also be considered a quasi-identifier and any exception should be justified. No justification is given.

\begin{table}[h]
  \setlength\belowcaptionskip{-1\baselineskip}
  \begin{center}
  \begin{tabular}{@{}lrrrr@{}} \toprule
  & \multicolumn{2}{c}{$\EA$} & \multicolumn{2}{c}{$\UEA$} \\
  \cmidrule(lr){2-3} \cmidrule(l){4-5}
  Aux info & $=1$ & $<5$ & $=1$ & $<5$ \\ \midrule
  $Q_*$ & 245 & 753 & 245 & 753 \\
  $\Qall$ & 33,925 & 73,136 & 9,125 & 22,491 \\
  $\Qclassmate$ & 120 & 216  & 120 & 216  \\
  & (1.7\%) &(3.0\%) &(1.7\%) &  (3.0\%) \\
  \addlinespace
  $\Qfriend$ & 31,797 & 69,543 & 7,108 & 19,203 \\
  $\QfriendEdu$ & 41,666 & 98,201 & 7,512 & 20,402 \\
  $\QLI$ & 5,542  & 10,939  & 732  &  2,310  \\
   & (34.2\%) &  (67.4\%) &(4.5\%) &  (14.2\%) \\
  \bottomrule
  \end{tabular}
\end{center}
      \begin{center}
    \caption{\label{table:results-k-anon}
    Number of students by effective anonymity ($\EA$) or ambiguous effective anonymity ($\UEA$) with respect to various choices of attacker auxiliary information ($Q_*$, $\Qall$, etc.), as described in this section.
    Numbers in parentheses are the value as a percentage of the relevant subset of the full dataset: for $\Qclassmate$, the 7,251 students with at least one forum post; for $\QLI$, the 16,224 students with at least one certificate. $\EA = \UEA$ for $Q_*$, $\Qclassmate$. }
        \end{center}
\end{table}

\subsubsection{Unambiguous uniques in the EdX dataset}
A naive interpretation of the 7.1\% unique students  is that an attacker who knows $\Qall$ would be able to definitively learn the grades of 7.1\% of the students. But there is a major source of ambiguity: missing information.
Gender, year of birth, and level of education were voluntarily self-reported by students.
Many students chose not to provide this information: 14.9\% of students records are missing at least one of these attributes.
It is missing in the raw data, not just the published data.
Thus, a female Italian born in 1986 might appear in the dataset with any or all three attributes missing.

This makes the 7.1\% result difficult to interpret.
From an inferential standpoint, the relevant question is not how many students have unique quasi-identifiers, but how many are \emph{unambiguously} unique. We compute the ambiguous effective anonymity $\UEA$ (defined in App.~\ref{app:UAE}) of each record by treating any missing attribute values as the set of all possible values for that attribute. This number may be much lower than 7.1\%.
We stress that this ambiguity comes from missing data, not from $k$-anonymity.

We find that 1.9\% of students (9,125 students) are unambiguously unique with respect to $\Qall$ and 4.7\% have ambiguous effective anonymity less than 5. Over 9,000 students are unambiguously identifiable in the dataset to anybody who knows all the quasi-identifiers, without knowing whether the students chose to self-report their gender, year of birth, or level of education. This allows an attacker to draw meaningful inferences about them.

\subsubsection{Limiting the attacker's knowledge}
Students in the EdX dataset are vulnerable to reidentification by adversaries who have much less auxiliary information than $\Qall$. We consider the (ambiguous) effective anonymity for three attackers who could plausibly reidentify students in the EdX dataset: a prospective employer, a casual acquaintance, and an EdX classmate. The results are summarized in Table~\ref{table:results-k-anon}.

In Section~\ref{sec:linkedin}, we carry out the prospective employer attack using LinkedIn. This demonstrates that some students in the EdX dataset can be reidentified by anybody.

\paragraph{Prospective employer}
Consider a prospective employer who is interested in discovering whether a job applicant failed an EdX course.
An applicant is likely to list EdX certificates on their resume.
The employer very likely knows $\QLI = $ \{gender, year of birth, location, level of education, certificates earned in courses 1--16\}.
$\QLI$ only includes those certificates actually earned, but omits courses in which a student enrolled but did not earn a certificate.

5,546 students in $\EdX$ have effective anonymity 1 with respect to $\QLI$, and 10,942 have effective anonymity less than 5.
These numbers may seem small, but they constitute 34.2\% and 67.4\% of the 16,224 students in the dataset that earned any certificates whatsoever.
Moreover, 732 students are unambiguously unique---333 of whom failed at least one course, and 38 of whom failed three or more courses. Thus, \emph{2.1\% of students (333 students) who earned certificates of completion failed at least one course and have unambiguous effective anonymity 1} with respect to $\QLI$.

\paragraph{Casual acquaintance}
Casual acquaintances might, in the course of normal conversation, discuss their experiences on EdX. They would likely discuss which courses they  took, and would naturally know each other's ages, genders, and locations.
So acquaintances know $\Qfriend = $\{gender, year of birth, location, enrollment in courses 1--16\} $\subseteq \Qall$.
6.7\% of students in $\EdX$ have effective anonymity 1 with respect to $\Qfriend$, and 14.6\% have effective anonymity less than 5.

Moreover, acquaintances typically know each other's level of education too, even though this is not included in $\Qall$. If we augment the acquaintance's knowledge with level of education $\QfriendEdu = \Qfriend \cup \{\mbox{education}\}$, then things become even worse.
8.7\% students in $\EdX$ have effective anonymity 1 with respect to $\QfriendEdu$, and 20.6\% have effective anonymity less than 5.

\paragraph{EdX classmate}
Each EdX course had an online forum for student discussions. Because these posts were public to all students enrolled in a given course, the number of forum posts made by any user was deemed publicly available information. But ignoring composition, EdX  did not consider the combination of forum post counts made by a user across courses.

Consider an attacker who knows $\Qclassmate = $\{number of forum posts in courses 1--16\} $\subseteq \Qall$.
120 students in $\EdX$ are unambiguously unique with respect to $\Qclassmate$, and 216 have ambiguous effective anonymity less than 5.
These numbers may seem minute, but they constitute 1.7\% and 3.0\% of the 7251 students in the dataset that made any forum posts whatsoever.
Effective anonymity and ambiguous effective anonymity are always the same for this attacker because $\Qclassmate$ excludes the demographic columns that are missing many entries.

Who knows $\Qclassmate$? 20 students \emph{in the dataset itself} enrolled in all 16 courses and could have compiled forum post counts across all courses for all other EdX students. To any one of these 20 students the 120 students with distinguishing forum posts are uniquely identifiable. Such an attacker can then learn these 120 students ages, genders, level of educations, locations, and their grades in the class.

In fact, each of the 120 vulnerable students can be unambiguously uniquely distinguished by 23--70 classmates; 60 students by 40--49 classmates each. This enables more classmates to act as attackers than just the 20 who took all courses.
This is because distinguishing a student using forum posts doesn't require being enrolled in all 16 courses. For each of the 120 vulnerable students, we find which subsets of their forum posts suffices to distinguish them. This analysis amounts to checking whether these students remain unambiguously unique if some subset of their forum post counts are redacted.

\subsection{Reidentifying EdX students on LinkedIn}
\label{sec:linkedin}

On LinkedIn.com, people show off the courses they completed. They may be unwittingly revealing which courses they gave up on.
2.1\% of students who earned certificates of completion (333 students) failed at least one course and have unambiguous effective anonymity 1 with respect to $\QLI$.

We reidentified three of these 333 students, with a rough confidence estimate of 90--95\%.

\subsubsection{Method}
People routinely post $\QLI$ on LinkedIn where it is easily searchable and accessible for a small fee.
We paid \$119.95 for a 1 month Recruiter Lite subscription to LinkedIn. Recruiter Lite provides access to limited search tools along with the ability to view profiles in one's ``extended network'': 3rd degree connections to the account holder on the LinkedIn social network. It is also possible to view public profiles outside one's extended network with a direct link, for example from a Google search.
A real attacker could build a larger extended network or pay for a more powerful Recruiter account.

We performed the attack as follows. We restricted our attention to 135 students in $\EdX$ who were unambiguously unique using only certificates earned plus at most one of gender, year of birth, and location, and who also had no missing demographic attributes.
We manually searched for LinkedIn users that listed matching course certificates on their profile by searching for course numbers (e.g., "HarvardX/CS50x/2012").
We attempted to access the profiles for the resulting users, whether they were in our extended network or by searching on Google.
If successful, we checked whether the LinkedIn user lists exactly the same certificates as the EdX student, and whether the demographic information on LinkedIn was consistent with the EdX student.
If everything matched, we consider this a reidentification.

\subsubsection{Results}
We reidentified 3 of the attempted 135 EdX students, each of whom registered for but failed to complete an EdX course. Two were unambiguously unique using only certificates of completion.
In each case, the EdX student's gender matched the LinkedIn user's presenting gender based on profile picture and name. In each case, the LinkedIn user's highest completed degree in 2013 matched the EdX student's listed level of education.
\begin{enumerate}[leftmargin=*]
\item Student 1's EdX record lists location $\ell_1$ and year of birth as $y_1$. The matching LinkedIn user began a bachelors degree in year $y_1+20$ and was employed in country $\ell_1$ in 2013.

\item Student 2's EdX record lists location $\ell_2$ and year of birth as $y_2$. The matching LinkedIn user began a bachelors degree in year $y_1+18$ and was in country $\ell_2$ for at  part of 2013.

\item Student 3's EdX record lists location $\ell_3$ and year of birth as $y_3$. The matching LinkedIn user graduated high school in year $y_3 + 19$, attended high school and currently works in country $\ell_3$. In 2013 the LinkedIn user was employed by an international firm with offices in $\ell_3$ and other countries.
\end{enumerate}

\subsubsection{Confidence}

We cannot know for sure whether our purported reidentifications on LinkedIn are correct because were instructed by our IRB not to contact the reidentified EdX students.

In this section, we estimate that our  reidentifications are correct with 90--95\% confidence. Moreover, an error is most likely a result of our imperfect ability to corroborate location and year of birth on LinkedIn, not a result of the protection afforded by $k$-anonymity.
Our analysis is necessarily very rough. A precise error analysis is impossible. We omit details to avoid imparting any other impression.

We consider two main sources of uncertainty.
First is the limited information available on LinkedIn profiles, especially age and location. We inferred a range of possible ages by extrapolating from educational milestones. We inferred a set of possible locations based on listed activities around 2013.
Both methods are imperfect.
The locations in EdX were inferred from IP address and are likely imperfect.
LinkedIn users or EdX students can report their attributes inconsistently. Note that they cannot lie about earning EdX certificates: this data comes from EdX itself and the LinkedIn certificates are digitally signed and cryptographically verifiable.\footnote{%
  An example certificate is available here: \url{https://verify.edx.org/cert/26121b8dec124bc094d324f51b70e506}.
  Instructions for verifying the signature are here: \url{https://verify.edx.org/cert/26121b8dec124bc094d324f51b70e506/verify.html}
}
We estimate the probability of error on at least one attribute inferred from LinkedIn is on the order of 5--10\%.

The second source of error is suppressed student records.
Of the 597,692 students enrolled in EdX courses over the relevant period, only 476,532 appear in the published dataset. 121,160 students (20.3\%) are completely suppressed.
We matched students $\x_{\mathsf{edx}}$ in EdX with users on LinkedIn $\x_{\mathsf{li}}$ using $\QLI =$ \{gender, year of birth, location, level of education, certificates earned in courses 1--16\} as well as we could.
An error will occur if a suppressed student $\x'_{\mathsf{edx}}$ is the true match for the LinkedIn user. For this to happen,  $\x'_{\mathsf{edx}}$ and $\x_{\mathsf{edx}}$ must agree on $\Qall$. If $\x_{\mathsf{edx}}$ is unique in the complete dataset, no error occurs.

We do a back of the envelope calculation of the chance of error from record suppression under two simplifying assumptions. First, that random student records are suppressed.\footnote{%
  There are more sophisticated techniques for estimating the probability of error under this assumption \cite{rocher2019estimating}. But in EdX omitted records are ``outliers  and  highly  active  users  because these users are  more  likely  to  be  unique  and  therefore  easy  to  re-identify'' \cite{edX-dataset}. As such, using the more sophisticated techniques would not give more meaning to our very coarse estimates.}
Second, that the number of certificates of completion that a user earns is statistically independent of their other attributes (assuming they registered for enough courses). We compute 99.5\%-confidence upper bounds for two parameters: the probability $p$ that a random EdX student matches our reidentified EdX student; the probability $q$ that a random EdX student earns the  same number of certificates as our reidentified EdX student. 121,160$\cdot pq$ is a very coarse estimate of the probability that a supressed student record causes an error. For the three students we reidentified, this comes out to 0.1--1\%.

A much less likely source of error is suppression of individual courses from a student's record. Such an error will occur if some courses for the purported match $\x_{\mathsf{edx}}$ were suppressed, and there is some other EdX student $\x'_{\mathsf{edx}}$ that is the true match for the LinkedIn user $\x_{\mathsf{li}}$.
This requires course suppression in  $\x_{\mathsf{edx}}$ and also $\x'_{\mathsf{edx}}$ (because $\x_{\mathsf{edx}}$ was unambiguously unique on $\QLI$ in the published EdX data).
All in all, we consider course suppression to be a much less likely source of error than student suppression or imperfect attribute inference on LinkedIn.

\subsection{EdX was ``properly'' deidentified}
El Emam, et al., criticize prior reidentification studies as using data that were ``improperly deidentified'' because they did not ``follow[] existing standards'' \cite{ElEmamSystematic}. They thus conclude that there is no convincing evidence of real-world failure of QI-deidentification techniques in a regulated context.

In contrast, the EdX dataset incontrovertibly followed existing standards. FERPA is the relevant regulation. It requires the published information to not enable identification of any student with reasonable certainty. The EdX dataset was specifically created to comply with FERPA, following Department of Education guidance and overseen by general council for Harvard and MIT \cite{edX-how-to}.

Moreover, the EdX dataset arguably followed the HIPAA Expert Determination--the standard used by El Emam, et al. The Expert Determination standard requires three things:\footnote{\url{https://www.hhs.gov/hipaa/for-professionals/privacy/special-topics/de-identification/index.html}} (1) Deidentification be performed by ``a person with appropriate knowledge \dots and experience''. (2) The person determines that the risk of reidentification is ``very small''. (3) The person ``documents the methods and results of the analysis that justify such determination.'' The creation of the EdX data was overseen by Harvard professors in computer science and statistics with specific expertise in privacy and inference. They find a ``low probability that the dataset will be re-identified'' and their methods and analysis are well-documented~\cite{edX-dataset}. The main deviation from the Expert Determination standard is the difference between ``very small'' and ``low'' reidentification risk.

\section{Conclusions}
\label{sec:conclusion}

In short, we show that $k$-anonymity -- and QI-deidentification generally -- fails \emph{on its own terms}.
Our attacks rebut three primary arguments that QI-deidentification's practioners make to justify its continued use.
First, we reidentify individuals in EdX dataset; it was ``properly de-identified'' by a ``statistical experts'' and in accordance with procedures outlined in regulation, meeting the  high bar set by El Emam, et al.~\cite{ElEmamSystematic}.
Second, our downcoding attacks demonstrate that even if every attribute is treated as quasi-identifying, $k$-anonymity and its refinements may provide no protection. Ours are the first attacks in either of these two settings.
Third, our attacks also undermine the claim that QI-deidentification meets regulatory standards for deidentification. The compound PSO and reidentification attacks challenge $k$-anonymity's status under GDPR and FERPA respectively.

Moreover, our attacks show that QI-deidentification violates three properties of a worthwhile privacy notion, even in practice. Namely, avoiding distributional assumptions, robustness against post-processing, and smooth degradation under composition. We expand on these next.

Downcoding attacks prove that whatever privacy is provided by QI-deidentification crucially depends on unstated assumptions on the data distribution.
One possible pushback is that our downcoding attacks use specially constructed distributions and hierarchies, not naturally occurring ones. But even a contrived counterexample proves that there is some unnoticed distributional assumption that is critical for security.
Moreover, the distributions and hierarchies in Theorem~\ref{thm:weaker-downcoding} are not so unnatural when considering that data are made, not found (to quote danah boyd).
Say an analyst wants to $k$-anonymize a high dimensional dataset. One natural approach is to find a low-dimensional projection with clusters of about $k$ rows each, and then construct the generalization hierarchy over this representation. The result could easily satisfy conditions that enable our downcoding attack or a direct extension.

Robustness against post-processing requires that further processing of the output, without access to the data, should not diminish privacy.
Downcoding proves that QI-deidentification is not robust to post-processing.
Our attacks recover specific secret information about a large fraction of a dataset's entries with probability close to 1.
Also, the EdX dataset also proves that $k$-anonymity is not robust to post-processing for purely syntactic reasons. The result of removing rows from a $k$-anonymous dataset may not satisfy $k$-anonymity as defined.
We see this in the EdX data: it is not in fact $5$-anonymous with respect to quasi-identifier $Q_*$ (courses), despite claims otherwise.
This fragility to post-processing is not so much a privacy failure as a syntactic weakness of the definition itself.

Smooth degradation under composition requires that a combination of two or more private applications mechanisms should also be private, albeit with worse parameters.
The EdX dataset proves that QI-deidentification is not robust to composition, even when done by experts in accordance with strict privacy regulations.
Ganta et al.\ present theoretical {composition attacks}, showing that if the same dataset is $k$-anonymized with different quasi-identifiers the original data can be recovered \cite{ganta2008composition}.
With the EdX dataset the possibility became reality.
To the best of our knowledge, this is the first example of such a failure in practice.

\medskip
The most important open question raised by this work is to characterize the power of downcoding attacks. What properties of a data distribution and generalization hierarchy enable downcoding? Is vulnerability to downcoding testable? In what settings is downcoding provably impossible? Can one demonstrated downcoding in the wild? We leave these questions for future work.

\paragraph{Responsible disclosure and data availability}
After reidentifying one EdX student, we reported the vulnerability to Harvard and MIT who promptly replaced the dataset with a heavily redacted one. Our IRB determined that this research was not human subjects research and did not need IRB approval. However, we were instructed by the IRB not to contact the reidentified LinkedIn users.
The code used in our analysis of the EdX dataset is at \url{https://github.com/a785236/EdX-LinkedIn-Reidentification}, but we do not distribute the dataset itself to protect the students' privacy.

\ifUSENIX
  \bibliographystyle{plain}
\else
  \bibliographystyle{alpha}
\fi
\bibliography{references}

\begin{thebibliography}{10}

\bibitem{AbowdDeclaration}
John~M. Abowd.
\newblock Supplemental {D}eclaration, 2021.
\newblock {State of Alabama v. US Department of Commerce}.

\bibitem{altman2020hybrid}
Micah Altman, Aloni Cohen, Kobbi Nissim, and Alexandra Wood.
\newblock What a hybrid legal-technical analysis teaches us about privacy
  regulation: The case of singling out.
\newblock {\em BUJ Sci. \& Tech. L.}, 27:1, 2021.

\bibitem{edX-how-to}
Olivia Angiuli, Joe Blitzstein, and Jim Waldo.
\newblock How to de-identify your data.
\newblock {\em Communications of the ACM}, 58(12):48--55, 2015.

\bibitem{barbaro_zeller_2006}
Michael Barbaro and Tom Zeller.
\newblock A face is exposed for {AOL} searcher no. 4417749.
\newblock {\em New York Times}, Aug 2006.

\bibitem{ontario2014big}
Ann Cavoukian and Daniel Castro.
\newblock {\em Big data and innovation, setting the record straight:
  de-identification does work}.
\newblock Information and Privacy Commissioner, Ontario, 2014.

\bibitem{cavoukian2014identification}
Ann Cavoukian and Khaled El~Emam.
\newblock {\em De-identification protocols: essential for protecting privacy}.
\newblock Information and Privacy Commissioner, Ontario, 2014.

\bibitem{CohenNissimSinglingOut}
Aloni Cohen and Kobbi Nissim.
\newblock Towards formalizing the gdpr’s notion of singling out.
\newblock {\em Proceedings of the National Academy of Sciences},
  117(15):8344--8352, 2020.

\bibitem{cormode2010minimizing}
Graham Cormode, Divesh Srivastava, Ninghui Li, and Tiancheng Li.
\newblock Minimizing minimality and maximizing utility: analyzing method-based
  attacks on anonymized data.
\newblock {\em Proceedings of the VLDB Endowment}, 3(1-2):1045--1056, 2010.

\bibitem{EdX_privacy}
Jon~P Daries, Justin Reich, Jim Waldo, Elise~M Young, Jonathan Whittinghill,
  Andrew~Dean Ho, Daniel~Thomas Seaton, and Isaac Chuang.
\newblock Privacy, anonymity, and big data in the social sciences.
\newblock 2014.

\bibitem{local-DP}
John~C Duchi, Michael~I Jordan, and Martin~J Wainwright.
\newblock Local privacy and statistical minimax rates.
\newblock In {\em 2013 IEEE 54th Annual Symposium on Foundations of Computer
  Science}, pages 429--438. IEEE, 2013.

\bibitem{DMNS06}
Cynthia Dwork, Frank McSherry, Kobbi Nissim, and Adam Smith.
\newblock Calibrating noise to sensitivity in private data analysis.
\newblock In {\em Theory of cryptography conference}, pages 265--284. Springer,
  2006.

\bibitem{ElEmamSystematic}
Khaled El~Emam, Elizabeth Jonker, Luk Arbuckle, and Bradley Malin.
\newblock A systematic review of re-identification attacks on health data.
\newblock {\em PloS one}, 6(12):e28071, 2011.

\bibitem{ganta2008composition}
Srivatsava~Ranjit Ganta, Shiva~Prasad Kasiviswanathan, and Adam Smith.
\newblock Composition attacks and auxiliary information in data privacy.
\newblock In {\em Proceedings of the 14th ACM SIGKDD international conference
  on Knowledge discovery and data mining}, pages 265--273. ACM, 2008.

\bibitem{k_anon_refinements_survey}
Olga Gkountouna.
\newblock A survey on privacy preservation methods, 2011.
\newblock \url{http://www.dblab.ece.ntua.gr/~olga/papers/olga\_tr11.pdf}.

\bibitem{edX_first_year}
Andrew Ho, Justin Reich, Sergiy Nesterko, Daniel Seaton, Tommy Mullaney, Jim
  Waldo, and Isaac Chuang.
\newblock {HarvardX and MITx}: The first year of open online courses, fall
  2012-summer 2013.
\newblock 2014.

\bibitem{pate-gan}
James Jordon, Jinsung Yoon, and Mihaela Van Der~Schaar.
\newblock Pate-gan: Generating synthetic data with differential privacy
  guarantees.
\newblock In {\em International conference on learning representations}, 2018.

\bibitem{t-closeness}
Ninghui Li, Tiancheng Li, and Suresh Venkatasubramanian.
\newblock t-closeness: Privacy beyond k-anonymity and l-diversity.
\newblock In {\em 2007 IEEE 23rd International Conference on Data Engineering},
  pages 106--115. IEEE, 2007.

\bibitem{MKGV07}
Ashwin Machanavajjhala, Daniel Kifer, Johannes Gehrke, and Muthuramakrishnan
  Venkitasubramaniam.
\newblock \emph{L}-diversity: Privacy beyond \emph{k}-anonymity.
\newblock {\em {TKDD}}, 1(1):3, 2007.

\bibitem{edX-dataset}
MITx and HarvardX.
\newblock {HarvardX-MITx Person-Course Academic Year 2013 De-Identified
  dataset, version 2.0}, 2014.

\bibitem{narayanan2008robust}
Arvind Narayanan and Vitaly Shmatikov.
\newblock Robust de-anonymization of large sparse datasets.
\newblock In {\em IEEE Symposium on Security and Privacy}, 2008.

\bibitem{myths-and-fallacies}
Arvind Narayanan and Vitaly Shmatikov.
\newblock Myths and fallacies of "personally identifiable information".
\newblock {\em Commun. ACM}, 53(6):24--26, June 2010.

\bibitem{BrokenPromises}
Paul Ohm.
\newblock Broken promises of privacy: Responding to the surprising failure of
  anonymization.
\newblock {\em UCLA Law Review}, 57:1701--1777, 2010.

\bibitem{wp-anonymisation}
Article 29 Data Protection~Working Party.
\newblock {\em Opinion 05/2014 on Anonymisation Techniques}.

\bibitem{rocher2019estimating}
Luc Rocher, Julien~M Hendrickx, and Yves-Alexandre De~Montjoye.
\newblock Estimating the success of re-identifications in incomplete datasets
  using generative models.
\newblock {\em Nature communications}, 10(1):1--9, 2019.

\bibitem{samarati2001protecting}
Pierangela Samarati.
\newblock Protecting respondents identities in microdata release.
\newblock {\em IEEE transactions on Knowledge and Data Engineering},
  13(6):1010--1027, 2001.

\bibitem{SamaratiS98}
Pierangela Samarati and Latanya Sweeney.
\newblock Generalizing data to provide anonymity when disclosing information
  (abstract).
\newblock In Alberto~O. Mendelzon and Jan Paredaens, editors, {\em Proceedings
  of the Seventeenth {ACM} {SIGACT-SIGMOD-SIGART} Symposium on Principles of
  Database Systems}. {ACM} Press, 1998.

\bibitem{sweeney2002k}
Latanya Sweeney.
\newblock K-anonymity: A model for protecting privacy.
\newblock {\em International Journal of Uncertainty, Fuzziness and
  Knowledge-Based Systems}, 10(05):557--570, 2002.

\bibitem{wong2009anonymization}
Raymond Chi-Wing Wong, Ada Wai-Chee Fu, Ke~Wang, and Jian Pei.
\newblock Anonymization-based attacks in privacy-preserving data publishing.
\newblock {\em ACM Transactions on Database Systems (TODS)}, 34(2):1--46, 2009.

\end{thebibliography}

\appendix

\section{Additional background on the EdX dataset}
We summarize the EdX dataset---the chosen quasi-identifiers, the implementation of $k$-anonymization, and the resulting published dataset $\EdXraw$ based on documentation included with the dataset \cite{edX-dataset} and in two articles describing the creation of the dataset itself \cite{EdX_privacy,edX-how-to}.

The raw dataset consisted of 841,687 rows for 597,692 students. Each row corresponded to the registration of a single student in a single course and included the information described above.
IP addresses were used to infer a student's location even when a student chose not to self-report their location.
EdX considered username and IP address to be identifying. Each username was replaced by a unique 7-digit identification number (UID). A username appearing in multiple rows was replaced by the same UID in each. IP addresses were redacted.

The dataset, with a row corresponding to a student-course pair, was $k$-anonymized according to two different quasi-identifiers $Q$ and $Q_*$ (each a subset of the attributes).
$Q =$  \{gender,  year of birth, country, course, number of forum posts\}. ``The  last  one  was  chosen  as  a  quasi‐identifier  because  the  EdX  forums are  somewhat  publicly  accessible  and  someone  wishing  to  re-­identify  the  dataset could, with  some  effort,  compile  the  count  of  posts  to  the  forum  by  username'' \cite{edX-dataset}.
Separately, the set of courses that each student enrolled in were considered to form a quasi-identifier: $Q' =$ \{enrolled in course 1, \dots, enrolled in course 16\}.
The data was $k$-anonymized first according to $Q_*$ and then according to $Q$. Additionally, $\ell$-diversity was enforced for the final course grade, with $\ell = 2$.

After aggregating the rows by UID, $\EdX$ can be seen as $k$-anonymized with respect to 17 overlapping quasi-identifiers: $Q_*$ as before and $Q_1,\dots,Q_{16}$, where $Q_i =$ \{gender,  year of birth, country, enrolled in course $i$,  number of forum posts in course $i$\}.

The final published result $\EdXraw$ includes 641,138 course registrations by 476,532 students across 16 courses.

As published, the EdX dataset had 641,138 rows, each representing to a single course registration for one of 476,532 distinct students.
But the object of our privacy concerns is a student, not a student-course pair. We aggregated the rows corresponding to the same UID. We call the result $\EdX$.

The creators of the EdX dataset failed to identify which attributes are publicly available---the very thing that experts are supposed to be good at. Specifically, level of education and certificates of course completion are not included in any of the quasi-identifiers despite both being readily available on LinkedIn. The exclusion of certificates is particularly indefensible: at the same time as the EdX dataset was being created, EdX and LinkedIn collaborated to allows LinkedIn users to include {cryptographically unforgeable} certificates of completion on their profiles.

\footnotetext{Different rows of the same student often listed different countries. Almost always there were only two different values, one of which was ``Unknown/Other.'' In this case, we used the other value for the student's unified record. In all other cases, we used ``Unknown/Other.''}

\subsection{Ambiguous effective anonymity}
\label{app:UAE}
We consider a relaxation of the notion of effective anonymity which we call \emph{ambiguous effective anonymity}. Let $I_\ambig(\Y,\y,Q) \triangleq \{n: \y_n(Q) \cap \y(Q)\neq \emptyset\}$. The ambiguous effective anonymity of $\y$ in $\Y$ with respect to $Q$ is $\UEA(\Y,\y,Q) = |I_\ambig(\Y,\y,Q)|$.
Ambiguous effective anonymity helps us reason about what an attacker can infer from the dataset.

\begin{definition}[Unambiguous uniqueness]
	We say $\y\in\Y$ is \emph{unique} with respect to $Q$ if $\EA(\Y,\y,Q) = 1$, and \emph{unambiguously unique} if $\UEA(\Y,\y,Q) = 1$.
\end{definition}

$\UEA$ is never less than $\EA$, and is very often greater in EdX. The presence of unambiguously unique records in a supposedly-anonymized dataset indicates a clear failure of syntactic anonymity.
Considering ambiguous effective anonymity makes critiquing $k$-anonymity much harder. We are giving $k$-anonymity the benefit of all the additional ambiguity that comes from missing data rather than from the anonymizer itself.

\ifUSENIX
\section{Deferred Proofs}
\label{app:proofs}

\subsection{Theorems~\ref{thm:weaker-PSO} and~\ref{thm:full-PSO}}

\begin{proof}[Proof outline]
Proofs for both compound PSO attacks follow the same general structure, using the corresponding downcoding attacks in non-black-box ways.
The compound PSO adversary $\A$ gets as input $\Y \gets M_H(\X)$.
It emulates the appropriate downcoding adversary, which produces an output $\Z$ such that $\X \preceq \Z \prec \Y$.

$\Z$ contains special generalized records $\z^t$ indexed by some $t \in [T]$ (equations~\eqref{eq:z-weak} and \eqref{eq:z-full}), and may also contain other records. $\A$ outputs $\Psi = \{\psi^t\}$ where $\psi^t:\x \mapsto \indic{\x \subseteq \z^t}$.

To complete the proof, one must show that the following hold with probability at least $1-\alpha(N)$:
\begin{itemize}
	\item $\psi(\X) =1/N$
	\item $(\psi\wedge \psi')(U) = 0$
	\item $\psi(U) < \negl(N)$
	\item $|\Psi| \ge L$
\end{itemize}
The first three are implied by the following:
\begin{itemize}
\item  $\forall t$, there exists a unique $\x^t \in \X$ such that $\x^t \subseteq \z^t$.
\item  $\forall t\neq t'$, $\z^t \cap \z^{t'} = \emptyset$.
\item  $\forall t$, $\x \sim \Att^D$, $\Pr_{\x}[\x \subset \z^t] < \negl(N)$.
\end{itemize}
For the downcoding attack from Theorem~\ref{thm:weaker-downcoding}, these properties are immediate. For the downcoding attack from Theorem~\ref{thm:full-downcoding}, the first two are immediate and the third follows from Claim~\ref{claim:well-spread}.

The requirements on $L$ and $\alpha$ follow from the parameters of the corresponding downcoding attacks. The downcoding attack for Theorem~\ref{thm:weaker-downcoding} yields $\Omega(N)$ records $\z^t$ with probability $1-\negl(N)$.
The downcoding attack for Theorem~\ref{thm:full-downcoding} yields $N$ records $\z^t$ with probability $1-\alpha(N)$.
\end{proof}

\fi

\ifextra
\else
\end{document}